\documentclass[12pt, a4paper]{article}
\usepackage{times,amsmath,epsfig,graphicx,balance,multirow,amssymb,geometry,bm,diagbox,amsthm,xspace,xcolor}
\usepackage{enumerate}
\usepackage{subcaption}
\usepackage{enumitem}

\usepackage[ruled, vlined, linesnumbered]{algorithm2e}

\usepackage{authblk}

\providecommand{\eg}{\emph{e.g.},\xspace}
\newcommand{\ie}{\emph{i.e.},\xspace}
\newcommand{\aka}{\emph{a.k.a.},\xspace}
\newcommand{\wrt}{\emph{w.r.t.},\xspace}

\newcommand\figref[1]{Figure~\ref{#1}}

\newcommand\secref[1]{Section~\ref{#1}}
\newcommand\equref[1]{Equation~(\ref{#1})}

\newcommand\algoref[1]{Algo.~\ref{#1}}
\newcommand\lineref[1]{Line~\ref{#1}}
\newcommand\lemref[1]{Lemma~\ref{#1}}

\newcommand\expref[1]{Example~\ref{#1}}
\newcommand\defref[1]{Definition~\ref{#1}}

\newcommand{\fakeparagraph}[1]{\vspace{1mm}\noindent\textbf{#1.}}
\newtheorem{definition}{Definition}     %[section]
\newtheorem{example}{Example}     %[section]
\newtheorem{lemma}{Lemma}     %[section]
\newtheorem{theorem}{Theorem}     %[section]
\ifodd 0

\else

\fi

\begin{document}

\title{Efficient and Accurate Range Counting on
Privacy-preserving Spatial Data Federation}

\author[1]{Maocheng Li}
\author[1,2]{Yuxiang Zeng}
\author[1,3]{Lei Chen}
\affil[1]{The Hong Kong University of Science and Technology, Hong Kong, China
\{csmichael,leichen\}@cse.ust.hk}

\affil[2]{School of Computer Science and Engineering, Beihang University, China
turf1013@buaa.edu.cn}
\affil[3]{The Hong Kong University of Science and Technology (Guangzhou), Guangzhou,
China}

\date{}
\maketitle

\begin{abstract}

A spatial data federation is a collection of data owners (\eg a consortium of taxi companies), and collectively it could provide better location-based services (LBS). For example, car-hailing services over a spatial data federation allow end users to easily pick the best offers. We focus on the range counting queries, which are primitive operations in spatial databases but received little attention in related research, especially considering the privacy requirements from data owners, who are reluctant to disclose their proprietary data. We propose a grouping-based technical framework named FedGroup, which groups data owners without compromising privacy, and achieves superior query accuracy (up to 50\% improvement) as compared to directly applying existing privacy mechanisms achieving Differential Privacy (DP). Our experimental results also demonstrate that FedGroup runs orders-of-magnitude faster than traditional Secure Multiparty Computation (MPC) based method, and FedGroup even scales to millions of data owners, which is a common setting in the era of ubiquitous mobile devices. \looseness=-1
\footnotetext[1]{Copyright may be transferred without notice, after which this version no longer be accesiible.}
\end{abstract}
\section{Introduction}
\label{sec:intro}

A federated database system is a collection of multiple cooperating but autonomous database systems \cite{sheth1990federated}. It reflects the real-world situation that the entire database of customer records is usually divided among different companies, and each company owns a distinct share of the entire database. Federated database systems have drawn many research interests \cite{shi2021efficient,bater2018shrinkwrap,DBLP:journals/pvldb/BaterEEGKR17} and been deployed in real-world applications. For example, multiple hospitals participate in an alliance to collectively contribute their data for discovering new drugs \cite{DBLP:journals/pvldb/BaterEEGKR17}. 

In the specific domain of location-based services (LBS) systems, a \textit{spatial data federation} considers a federation of \textit{spatial data}, such as locations or trajectories. Similar to general-purposed federated databases, spatial data federation has also seen a wide range of real-world applications. For example, Alibaba's AMap \cite{technode2021} in China provides car-hailing services over a federation of different taxi companies, which enables users to have more flexibility and easily pick the best offers. 

To support the aforementioned LBS applications, query processing techniques are needed for a wide range of spatial queries. Among them, \textit{range counting queries} are one of the most important primitive operations. Range counting queries return the count of the spatial objects located within a given range. A car-hailing federated system may frequently issue the following range counting query: \textit{"how many cars are within 500 meters of the location of a customer?"} \looseness=-1

\begin{figure}[htbp]
    \centering
	\includegraphics[width=0.5\linewidth]{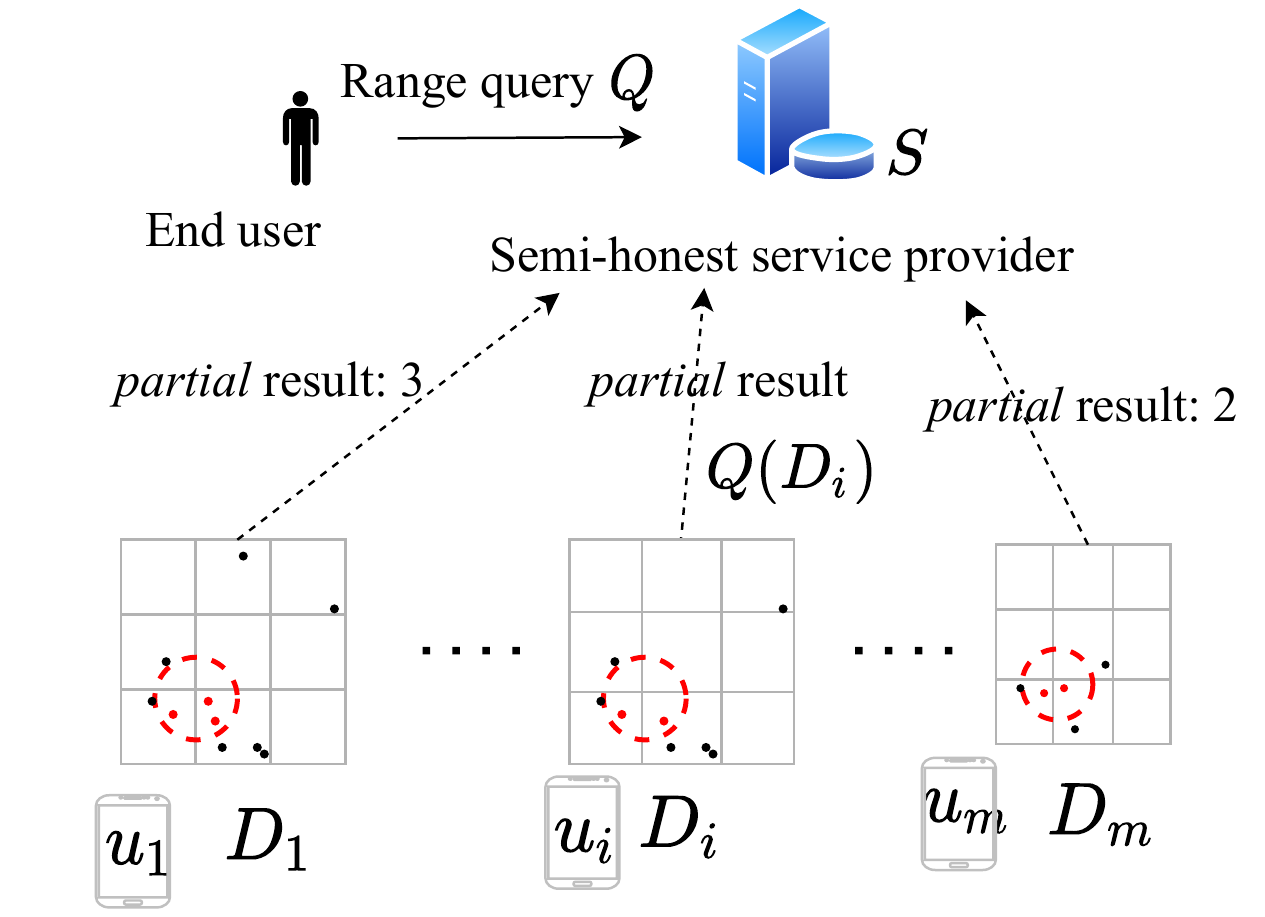}
	\caption{\small FPRC problem. Each one of $m$ data owners $u_i$ has a spatial database $D_i$. An end user issues a range counting query $Q$ over the spatial data federation: $D_1 \cup \ldots \cup D_m$.}
	\label{fig:problem}
\end{figure}

In this paper, we target at the \textit{\underline{F}ederated \underline{P}rivacy-preserving \underline{R}ange \underline{C}ounting (FPRC)} problem (\figref{fig:problem}), with a special focus on the large-scale spatial data federations, where the number of data owners is large (\eg more than 10K). Such large-scale settings have recently drawn more attentions due to the wide adoption of mobile devices. For example, in a decentralized crowdsourcing platform \cite{TongJoS17}, each data owner is a mobile device, and each user holds his own data and is not willing to disclose any unprotected private data to other parties. In such applications, the number of data owners could easily scale up to the order of millions. 

It is challenging to support range counting queries over a spatial data federation, because each data owner is often reluctant to disclose its proprietary data. Directly applying traditional Secure Multiparty Computation (MPC) techniques or Differential Privacy (DP) solutions would either be too computationally expensive, or bring too much noise to outweigh any meaningful results for the LBS applications. As demonstrated by our experiments in \secref{sec:experiment}, the MPC baseline is orders-of-magnitude slower than our solution, requiring more than 13 hours when the input size is large, while the DP baseline loses more than 50\% accuracy. Due to the interactive nature of the LBS applications, an ideal solution should offer excellent accuracy, practical efficiency, and proven privacy guarantee. 

To tackle the challenges of FPRC problem, we propose a grouping-based technical framework named FedGroup. To reduce the large amount of noise injected by directly applying DP for each data owner, we assign similar data owners (mobile users in our setting) into groups. We argue that injecting DP noise in each group (as opposed to injecting noise for each data owner) suffices to provide privacy guarantee. The reason is due to the fact that certain users may have social relationships with other data owners (\eg they belong to the same family), and they share similar trajectories and whereabouts. 

To achieve effective grouping, we adopt a commonly accepted concept named $(k,r)$-core in graph analysis \cite{DBLP:journals/pvldb/ZhangZQZL17} to ensure that the group has strong connections inside (here the nodes in the graph are data owners and the edges between nodes indicate that they have strong spatial similarities, \ie their trajectories are similar). We formulate an optimization problem to find the best way to put users into groups in order to minimize the error introduced by DP noise in the FPRC problem. We show that the problem is NP-hard and provide an efficient greedy-based algorithm with a performance guarantee. 

Last but not least, the grouping of users relies on how to measure the spatial similarity between users (\eg their trajectories). It is a non-trivial task to construct the similarity graph between each pair of users, considering that the trajectories are private information of each data owner. We devise a novel hybrid solution to combine DP and MPC to achieve accurate similarity graph construction. \looseness=-1

To summarize, we make the following contributions: 

\begin{itemize}[leftmargin=*]
	\setlength\itemsep{0.1em}

	\item We develop a novel technical framework FedGroup, to utilize an offline grouping step to largely reduce the amount of noise needed in the Federated Privacy-preserving Range Counting (FPRC) problem. We introduce the FPRC problem in \secref{sec:problem}. 
 
	\item In order to achieve effective grouping in FedGroup, we show that it is an NP-hard problem to achieve optimal grouping to minimize the amount of noise needed in FPRC problem. Then, we propose an efficient greedy-based algorithm with a performance guarantee to achieve effective grouping. 
	
	\item We also devise effective techniques to construct the spatial similarity graph between users by combining DP and MPC techniques, considering data owners' trajectories as private information. The details are presented in \secref{sec:fedgroup_details}. 
	
    \item We conduct extensive experiments to validate the effectiveness and efficiency of our proposed solution. The results are shown in \secref{sec:experiment}. 

\end{itemize}

In addition, we %introduce the background in \secref{sec:background}, 
review related works in \secref{sec:relatedWork} and conclude in \secref{sec:conclusion}.

\section{Problem Definition}
\label{sec:problem}
In this section, we introduce some basic concepts, the adversary model, and the definition of the Federated Privacy-preserving Range Counting (FPRC) problem. 
Due to page limitations, a toy example is provided in the appendix. 

\subsection{Basic Concepts}
\label{subsec:basic_concept}

\begin{definition}[Location]\label{def:location}
	A location $l=(x,y)$ represents a 2-dimensional spatial point with the coordinates $(x,y)$ on an Euclidean space. 
\end{definition}

Locations in an Euclidean space is commonly seen in existing work \cite{DBLP:conf/sigmod/SheT0S17,DBLP:conf/waim/GaoTSSCX16,DBLP:conf/dasfaa/TaoZZTC018}.

\begin{definition}[Data owner]\label{def:data_owner}
	There are $m$ data owners $u_1, \ldots, u_m$. Each data owner $u_i$ owns a spatial database $D_i$. Each spatial database $D_i$ consists of multiple data records (\textit{locations} as in \defref{def:location}) , \ie $D_i=\{l_1, \ldots, l_{|D_i|}\}$. 
\end{definition}

We also refer to each data owner's database $D_i$ as a \textit{data silo}. We let $D= \cup_{i=1}^m D_i$ denotes the collection of all data silos, \ie the union of all data records from each data silo. 

\begin{definition}[Range counting query]\label{def:range_count}
A range counting query $Q$ asks for how many data records are within distance $r$ to the 	query location $q_0$, \ie
\begin{equation}
		Q(r, q_0) = \sum_{l \in D} \mathbb{I}(d(l, q_0)<r)
\end{equation}
where $\mathbb{I}(\cdot)$ is the indicator function which equals to 1 if the predicate $d(l, q_0)<r$ is true, and 0 otherwise. $d(\cdot)$ is the Euclidean distance function. $l \in D$ is any record from the union of all data records from each data silo. 
\end{definition}

\subsection{Privacy and Adversary Model}
\label{subsec:adversary_model}

There are mainly three parties (roles) in our spatial data federation (see \figref{fig:problem}): (1) the \textbf{end user} who issues the query; (2) the \textbf{service provider} who receives the query and coordinates the execution of the query; and (3) each \textbf{data owner} $u_i$ who owns the private spatial database $D_i$. 

We assume that all parties are \textit{semi-honest}, meaning that they are \textit{curious but not malicious}. They are \textit{curious} about other parties' private information, but they honestly follow and execute system protocols. The setting has been widely adopted in recent privacy-preserving LBS related applications \cite{to2018,DBLP:conf/icde/TaoTZSC020,DBLP:journals/vldb/TongZZCS20}.

\subsection{Federated Privacy-Preserving Range Counting problem}
\label{subsec:problem}
Based on the previous concepts and the adversary model, we now define the Federated Privacy-preserving Range Counting (FPRC) problem as follows.

\begin{definition}[Federated Privacy-preserving Range Counting problem]\label{def:problem}
	Given a federation of $m$ spatial databases $D_1, \ldots, D_m$, a range counting query $Q(q_0, r)$, and a privacy parameter (\aka privacy budget) $\epsilon$, the FPRC problem asks for a query answer $\tilde{Q}$, with the following privacy requirements:
	\begin{itemize}[leftmargin=*]
		\item R1. The computed final result $\tilde{Q}$ satisfies $\epsilon$-Differential Privacy (DP), where the definition of neighboring databases refers to changing one single record in any data silo $D_i$.
		\item R2. The intermediate result disclosed by any data silo $D_i$ satisfies $\epsilon$-DP.  
		\item R3. The private inputs for each data silo $D_i$ are \textit{confidential} if there is any multiparty computation involved. 
	\end{itemize}
\end{definition}

Two baselines, each respectively based on MPC and DP, are proposed. The MPC baseline directly applies MPC technique to compute a secure summation over the query result from each data silo. The DP baseline injects an instance of Laplace noise to each data silo, and collects the aggregated noisy summation. 
% The details are deferred to the appendix. 

As mentioned in \secref{sec:intro}, the \textbf{key challenges} to solve the FPRC problem are threefold: 1) privacy: each data owner is reluctant to disclose its sensitive data; 2) efficiency: directly applying MPC results an impractical solution; 3) accuracy: since the number of data owners could reach the order of millions, the scale of total noise injected by the DP baseline is too large.

\section{Our Solution FedGroup}
\label{sec:fedgroup_details}

In this section, we introduce the technical details of FedGroup, which addresses the challenges of the FPRC problem. We highlights its overall workflow first (\figref{fig:fedgroup_workflow} is an illustrative figure), and then introduce the details in each of the steps: 1) constructing the spatial similarity graph; 2) finding groups given the similarity graph; and 3) partial answers and aggregation. 
Due to page limitations, please refer to the appendix for more examples and detailed proofs. 

\fakeparagraph{Key Idea and Intuition}
In our setting of the Federated Privacy-preserving Range Counting (FPRC) problem (as in \secref{sec:problem}), the number of data owners ($m$) is potentially very large (\eg millions), as each data owner could be a mobile user. The key idea of FedGroup is based on the following observation: certain mobile users have strong connections (due to social ties or other collaboration relationships) with each other. If we consider such pairs of data owners to belong to the same group, then the privacy protection inside the group could be relaxed. For example, there is no need to consider privacy protection of users' trajectories within a family, as family members are very likely to know about other members' whereabouts during the day. Thus, DP noise is only needed cross groups, as opposed to the case in the DP baseline, where an instance of Laplace noise is injected for every mobile user. As a result, the overall noise injected for the query result for the FPRC problem could be greatly reduced.

\begin{figure}[t!]\centering \vspace{-1ex}
	\scalebox{0.28}[0.28]{\includegraphics{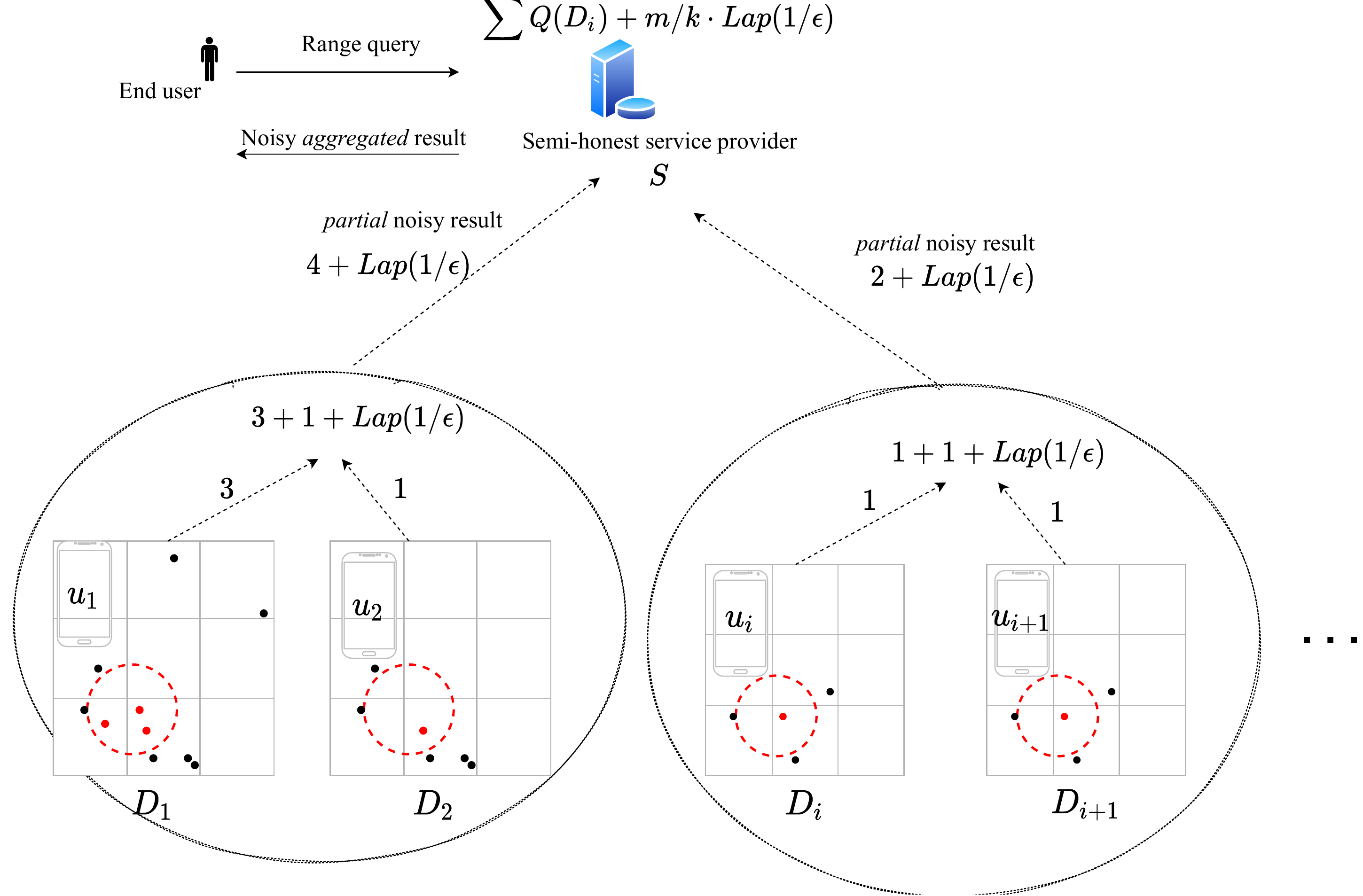}}
	\caption{\small FedGroup workflow. }

	\label{fig:fedgroup_workflow}
\end{figure}

\subsection{Spatial Similarity Graph Construction}
\label{subsec:sim_graph}

In this step, the goal is to construct an undirected graph $G_s=(V, E)$ between data owners. Each node in $V$ of the graph represents a data owner $u_i$, and an edge in $E$ between two nodes $u_i$ and $u_j$ indicates the similarity of the spatial databases $D_i$ and $D_j$ owned by $u_i$ and $u_j$, respectively. The weight of an edge measures the strength of the connection between $u_i$ and $u_j$, and we use cosine similarity as the weight of the edge. 

We first introduce the details of similarity graph, including how to measure the spatial similarity between data owners. Then, we consider the privacy-preserving setting, where each data owner's spatial database is considered private, and present our solutions of computing the similarity function between data owners in a differentially private manner. 

\fakeparagraph{Spatial similarity graph}
\label{subsubsec:spatial_sim_graph}
The key to construct the spatial similarity graph $G_s=(V,E)$ is measuring the weight of each undirected edge $e=(u_i, u_j)$ between two data owners $u_i,u_j\in V$. In the following section, we introduce how to compute the similarity function between the two data silos $D_i$ and $D_j$. 

For each data silo $D_i$, we use a grid structure $T$ to decompose the spatial domain into a number of grids, where $|T|$ denotes the number of grids. Then, within each grid, a count is computed to denote the number of spatial points (records) of $D_i$ falling inside the grid. Thus, we obtain a $|T|$-dimensional count vector  for the data silo $D_i$. 

Similarly, for data silo $D_j$, we could obtain another count vector $v_j=[c_{j,1}, c_{j,2}, \\ \ldots, c_{j,|T|}]$. Since the two vectors $v_i$ and $v_j$ are with the same length $|T|$, we could use \textit{cosine similarity} to measure their similarity as: 

\begin{equation}
  \text{sim}(D_i, D_j) = \cos(v_i, v_j) = v_i \cdot v_j / \|v_i\|\|v_j\|
  \label{eq:spatial_sim}
\end{equation}

The cosines similarity ranges from 0 to 1, and the larger the value is, the more similar the two count vectors are. In our application, the higher the cosine similarity between the two count vectors is, the more similar are the two associated spatial databases, owned by two different data owners.

\fakeparagraph{Constructing the graph} 
To construct the spatial similarity graph $G_s = (V, E)$, first we insert all the data owners as the node set $V$. Then, we iterate over all pairs of data owners $u_i, u_j \in V$, and measure the spatial similarity of their spatial databases $D_i$ and $D_j$, according to \equref{eq:spatial_sim}. The weight of the corresponding edge $e=(u_i, u_j)$ is set as $\text{sim}(D_i, D_j)$, \ie $w(e)=\text{sim}(D_i, D_j)$.

Without considering the privacy of each data silo $D_i$, constructing the spatial similarity graph $G_s$ is straightforward. % (as shown in \algoref{algo:construct_graph_plaintext}). 
The method is running on the service provider $S$. First, we create the node set $V$ of $G_s$ by inserting all data owners $u_1, \ldots, u_m$ into $V$. Then, the method simply iterates over the pairs of data owners $u_i,u_j$ for $i,j \in [m]$ and $i < j$, and calculates $\text{sim}(D_i, D_j)$. If the similarity is larger than the given threshold $r$, then an edge $(u_i, u_j)$ will be inserted into the edge set $E$ of $G_s$, with the edge weight set as $w(u_i, u_j)=\text{sim}(D_i, D_j)$. 

The time complexity is $O(m^2)$, since we iterate over all possible pairs of data owners. The space complexity required is also $O(m^2)$, because we need to store $G_s$ on $S$, including all the nodes and the edges (and their weights).

\fakeparagraph{Privacy-preserving computation}
Although it is straightforward to construct $G_s$ without considering the privacy of each data silo $D_i$, it is crucial to provide a proven privacy guarantee when we measure the spatial similarity between each pairs of data silos. The first and foremost motivation of this work is to consider privacy-preservation in the spatial federation setting, where each data owner's data silo $D_i$ is considered sensitive. 

Obviously, the non-private way of constructing the graph %\algoref{algo:construct_graph_plaintext} 
fails to meet the privacy requirement R2 in our problem definition in \defref{def:problem}. In R2, it requires that the intermediate results shared by each data owner satisfy $\epsilon$-DP. However, in the non-private version,  %\algoref{algo:construct_graph_plaintext}, 
the provided $D_i$ (or its count vectors $v_j$) are directly published by each data owner to the service provider $S$. Clearly it is a privacy failure. 
Furthermore, R1 in \defref{def:problem} may fail. It is not clear whether the constructed similarity graph $G_s$ provides any formal privacy guarantee \wrt changing of one record in any data silo $D_i$. 

Thus, we present a hybrid solution, which utilizes both Secure Multiparty Computation (MPC) and the standard Laplace mechanism in DP.

\fakeparagraph{Hybrid solution} The hybrid solution combines the advantages of both worlds -- the high accuracy of MPC (because the computations are \textit{exact}), and the high efficiency of DP (because each data silo performs the Laplace computation independently, and injecting Laplace noise itself is an efficient operation). 

The key idea of the hybrid solution is that: we use the efficient Laplace mechanism in DP to quickly compute a noisy weight $\tilde{w}$. Instead of directly using this noisy weight as a surrogate weight for the edges on the ground-truth graph, we only use it as a filtering mechanism. If the noisy weight $\tilde{w}$ is too small as compared to the given threshold $r$, we could quick conclude that the edge should \textit{not} be inserted to the graph. Or, if the noisy weight $\tilde{w}$ appears to be high, we quickly conclude that the edge should be inserted to the graph. If the noisy weight lies on the borderline, which shows uncertainty and may introduce errors, we invoke the computationally heavy MPC method to calculate the exact edge weight given the two private inputs from two data owners.

\algoref{algo:construct_graph_hybrid} shows the detailed steps. The inputs $\tilde{v_1}, \ldots, \tilde{v_m}$ are the noisy counts vectors, obtained by injecting Laplace noise to the counts vectors $v_1, \ldots, v_m$.

\begin{algorithm}[t]
	\DontPrintSemicolon
	\KwIn{$u_1, \ldots, u_m, \tilde{v_1}, \ldots, \tilde{v_m}, r, r_l, r_u$. }
	\KwOut{$\tilde{G_s}=(V, \tilde{E})$.}
    $V := \{u_1, \ldots, u_m\}$  \;
    $\tilde{E} := \{\}$ \;
	\ForEach{$u_i \in V$}{
    \ForEach{$u_j \in V$ and $j > i$} {
      $e := (u_i, u_j)$ \;
      $\tilde{w} := \cos(\tilde{v_i}, \tilde{v_j})$ \;
      \If{$e.\text{weight} < r_l$}{
        \label{line:hybrid_r_l}
        Continue \;
      }
      \ElseIf{$\tilde{w} > r_u$}{
        \label{line:hybrid_r_u}
        $e.\text{weight} := \tilde{w}$\;  
        Insert $e$ to $\tilde{E}$ \;
      }
      \Else{
        \label{line:hybrid_r_m}
        $e.\text{weight} :=$ Secure\_Sim($D_i, D_j$) \;
      \If{$e.\text{weight} > r$}{
        Insert $e$ to $\tilde{E}$ \;
       }
      }
    }

	}

	\Return{$\tilde{G_s}=(V, \tilde{E})$}\;
	\caption{\texttt{Construct $G_s$ - Hybrid solution} }
	\label{algo:construct_graph_hybrid}
\end{algorithm}

\fakeparagraph{Privacy analysis}
We show that our hybrid solution satisfies the privacy requirements as in \defref{def:problem}. 

\begin{theorem}
  The hybrid solution in \algoref{algo:construct_graph_hybrid} satisfies all privacy requirements (R1-R3) in \defref{def:problem}. 
\end{theorem}
\begin{proof}
  Since the only released intermediate results are the noisy counts $\tilde{v_i}$, R2 is satisfied. R1 is satisfied by the post-processing and parallel composition theorem of DP, because the output graph $\tilde{G_s}$ only depends on the noisy counts. R3 is satisfied because the multi-party computation only happens at \lineref{line:hybrid_r_m}, and the confidentiality of the data is provided by the MPC protocol. 
\end{proof}

\subsection{Finding Groups}
\label{subsec:grouping}

\fakeparagraph{$(k,r)$-core}
We adopt a commonly accepted concept $(k,r)$-core \cite{DBLP:journals/pvldb/ZhangZQZL17} from graph mining literature to define the groups. The concept of $(k,r)$-core considers two important criteria of determining closely related groups (\aka community detection): \textit{engagement} and \textit{similarity}, such that the group members not only have strong social connections with each other (strong engagement), but also share high similarities (\ie their trajectories are similar). 

In our setting, as we focus on the spatial database of data owners, rather than the social graph, we require that the groups we form demonstrate strong similarities between data owners. Thus, we require each group to be a $r$-clique, meaning each group is a clique, and the edge weight (\ie the spatial similarity, introduced in \secref{subsec:sim_graph}) between the data owners is larger or equal to $r$. %It may be interesting to adopt the original $(k,r)$-core idea to define the group, and it requires privacy-preserving techniques \wrt the social graphs of data owners. We defer this to future works. 

\begin{definition}[$r$-group]
  \label{def:r_group}
  A $r$-group $g$ is a clique in the spatial similarity graph $G_s$, such that there exists an edge between any two data owners belonging to the group, and the edge weight is larger than the threshold $r$, \ie 
  for any two data owners $u_i, u_j \in g$,
  \begin{equation}
		\exists e=(u_i, u_j) \in E\; , w(e)>r
	\end{equation}
\end{definition}

\fakeparagraph{Optimizing the grouping}
The concept of $r$-group gives meaning to grouping data owners together, because they are similar. However, there could be various ways of putting data owners into different groups. For example, a data owner could belong to multiple $r$-group, and which group should we choose to put the data owner in? 

In this section, we formulate an optimization problem to connect the utility goal of the ultimate problem of this paper -- the FPRC problem, with the grouping strategy. Then, we show that this optimization is an NP-hard problem, and we propose a greedy algorithm, which is effective and efficient. 

We first define the Data Owner Grouping (DOG) problem. 

\begin{definition}[Data Owner Grouping (DOG) problem]
  \label{def:dog_grouping_problem}
  Given the inputs to the FPRC problem in \defref{def:problem}, the DOG problem asks for a way to assign data owners into $\lambda$ disjoint groups: $g_1, g_2, \ldots, g_\lambda$, where each group $g_i$ is a $r$-group (according to \defref{def:r_group}), such that the aggregated amount of noise injected by FedGroup for the FPRC problem is minimized. 
\end{definition}

Next, we show that to minimize the error (the total amount of noise injected) for the FPRC problem, it is equivalent to minimize the number of groups $\lambda$. 

\begin{lemma}
  \label{lemma:minimize_lambda}
  Minimizing the error in the FPRC problem is equivalent to minimizing $\lambda$, which is the total number of groups in the DOG problem. 
\end{lemma}

\begin{proof}
  We review the error of our FedGroup solution in the FPRC problem. %As \secref{sec:fedgroup_overview} briefly shows, 
  FedGroup assigns data owners into disjoint groups, and each group injects one single instance of Laplace noise to the query result. Then, the partial noisy result is collected from each group and aggregated as the final answer. Thus, the total error of our final answer $\tilde{Q}$, which is measured by its variance, is: 
  \begin{equation}
    \text{Var}(\tilde{Q}) = \text{Var}(\sum_{i=1}^\lambda\text{Lap}(1/\epsilon)) = 
    \lambda\cdot \text{Var}(\text{Lap}(1/\epsilon)) = 2\lambda / \epsilon^2.   
    \label{eq:fedgroup_error}
  \end{equation}

  Thus, to minimize the total error in \equref{eq:fedgroup_error}, we need to minimize $\lambda$. 
\end{proof}

Now, our DOG problem becomes finding an optimal way of assigning data owners into disjoint groups, where each group is a $r$-clique, and the number of groups is minimized. In fact, the problem could be reduced from the Minimum Clique Partition (MCP) problem and the hardness result is shown as follows.

\begin{theorem}
  The DOG problem is NP-hard.
\end{theorem}

\begin{proof}

  We reduce the NP-hard problem, Minimum Clique Partition (MCP), to the DOG problem in the appendix. Thus, DOG shares the same hardness. 
\end{proof}

\fakeparagraph{Greedy algorithm} Since the DOG problem is an NP-hard problem, we devise an efficient greedy solution (\algoref{algo:greedy_grouping}) to obtain effective grouping and provide a bounded noise scale in the returned solution. The solution is inspired by the greedy solution to the Minimum Graph Coloring problem, which is shown as an equivalent problem to the MCP problem.

\begin{algorithm}[t]
	\DontPrintSemicolon
	\KwIn{$\tilde{G_s}=(V, \tilde{E})$. }
	\KwOut{$g_1, g_2, \ldots, g_\lambda$.}
    $G_c := $ ComplementGraph($\tilde{G_s}$)  \;  
    assigned $:= $ InitializeIntegerArray(size=$|V|$, initValue=0) \;      
    $\lambda :=$ 1 ;
    $\text{assigned}[u_1] = 1$ ;
    $g_1$.insert($u_1$) \;
    $u := $ NextUnassignedVertex($V$, assigned) \;
    %$u := $ UnassignedVertex($V$, assigned) \;
    \While{$u$ exists } {
      
      $gid :=$ NextGroupId($u$.neighbors(), assigned) \label{line:greedy_nextgroupid}\;
      $\text{assigned}[u] = gid$ ; 
      $g_{gid}$.insert($u$) \;
      \If{$gid > \lambda$} {
        $\lambda = gid$ \;
      }
      $u := $ NextUnassignedVertex($V$, assigned) \;
    }

	\Return{$g_1, g_2, \ldots, g_\lambda$}\;
	\caption{\texttt{Greedy Find Groups} }
	\label{algo:greedy_grouping}
\end{algorithm}

\fakeparagraph{Performance guarantee} In Theorem \ref{the:bounded-number}, we show the greedy solution has a bounded number of groups and offers a bounded noise for the FPRC problem. 

\begin{theorem}\label{the:bounded-number}
  If $d$ is the largest degree of a node in the complement graph $G_c$, then \algoref{algo:greedy_grouping} returns at most $d+1$ groups. 
\end{theorem}
\begin{proof}
  We focus on a data owner $u$ with degree $d$ (the maximum degree in the graph) as \algoref{algo:greedy_grouping} proceeds. There are at most $d$ neighbors of $u$ that we should avoid using the same group id. Since we are using the lowest-numbered group ids that have not been used by any of the neighbors, among group id $1, 2, \ldots, d+1$, there is at least one id that could be used by $u$ (\eg the first $d$ group ids are used, and we now use $d+1$). Thus, we conclude that at most $d+1$ groups are returned by \algoref{algo:greedy_grouping}. 
\end{proof}

\fakeparagraph{Time complexity} Each time when \algoref{algo:greedy_grouping} processes a vertex, it iterates over the neighbors of the vertex to find the lowest-numbered group id that is available. Thus, overall, it takes $O(|E|)$ time to run the algorithm. The graph complement step at the initialization also takes $O(|E|)$ time.

\subsection{Partial Answers and Aggregation}
\label{subsec:partial_aggregation}
\secref{subsec:grouping} describes the most critical step of FedGroup, which is finding the groups for data owners. After the grouping, data owners belonging to the same group could avoid injecting separate instances of Laplace to their individual query answers. Instead, each group aggregates the partial answers from data owners and injects one instance of Laplace noise to the partial aggregated answer. 

As the last step, the service provider $S$ collects the partial noisy answers from all groups $g_1, \ldots, g_\lambda$. The noisy answers are aggregated (taking a summation) and then returned to the end user. 

\subsection{Extension to Range Aggregation Queries}
In this paper, we focus on the most fundamental primitives for LBS systems, the range counting queries. However, the techniques we present could be easily extended to other aggregation queries, including range \texttt{SUM()} and \texttt{AVG()}. 

To extend FedGroup to \texttt{SUM()}, the important extension is the scale of noise injected for each group (the sensitivity). Since now each data record in a spatial database does not only affect the final query result by +/-1, we need to inject the worst-case scale to satisfy differential privacy. To avoid injecting unbounded Laplace noise, a truncation could be performed, \ie we truncate all data records to make a certain attribute smaller or equal to a truncation parameter $\theta$. 

The extension to \texttt{AVG()} is straightforward as the average could be calculated by \texttt{SUM()}/\texttt{COUNT()}.

\section{Experimental Study}
\label{sec:experiment}

In this section, we first introduce the experimental setup in Sec.~\ref{subsec:experimentSetup} and present the detailed experimental results in Sec.~\ref{subsec:experimentResults}. 

\subsection{Experimental Setup}
\label{subsec:experimentSetup}

\fakeparagraph{Datasets} We use a real-world geo-social network datasets Gowalla \cite{DBLP:conf/kdd/ChoML11} and a randomly generated synthetic dataset in our experiments. 

\fakeparagraph{Baselines} We compare our proposed FedGroup (short as \textbf{FG}) method with the DP baseline %(\secref{subsubsec:dp_baseline}
(short as \textbf{DP}) and the MPC baseline (short as \textbf{MPC}).

For FedGroup, we also implement a variant based on cliques sizes (a heuristic based solution). It lists the cliques in the graph in descending sizes, and keeps adding the largest clique into a new group until all nodes are assigned with a group. We name this variant as FedGroup-Exhaustive (short as \textbf{FG-Ex}). Note that finding the maximum clique in a graph is in general an NP-hard problem, so we expect that this variant is only tractable on small graphs. 

\fakeparagraph{Metrics}
\label{subsec:experimentsMetrics}
We focus on the following end-to-end metrics:
\begin{itemize}[leftmargin=*]
    \item Mean relative error (MRE): the ratio of error as compared to the true result, \ie $|$returned result - true result$|$/true result, averaged over repeated queries. 
    \item Mean absolute error (MAE): the absolute error of the returned result as compared to the true results, \ie $|$returned result - true result$|$, averaged over repeated queries. 
    \item Query evaluation time: the time to execute the spatial count queries in seconds, averaged over repeated queries. 
\end{itemize}

\fakeparagraph{Control variables}

\begin{itemize} [leftmargin=*]
  \item The number of data owners: $m \in [ \textbf{500}, 1K, 2K, 3K]$. For the scalability test: %, the synthetic dataset contains data owners in the range of 
  $m \in [100K, 250K, 500K, 1M]$. 
  \item Privacy budget:  $\epsilon \in [0.2, \textbf{0.3}, 0.4, 0.5]$. The privacy budget considered is relatively small, because it is a personal budget for each data owner, and repeated queries may result linear composition of the budget. 
\end{itemize}

\subsection{Experimental Results}
\label{subsec:experimentResults}

Here, we present the results on query accuracy to first verify the motivation of our work -- our proposed solution FedGroup offers superior query accuracy (up to 50\% improvement) as compared to the DP baseline. 
We also compare the running time of query execution to show that the MPC baseline is not tractable even in moderate data sizes, while our solution FedGroup offers excellent scalability. Overall, we have a better utility/efficiency trade-off than the MPC baseline.

 \fakeparagraph{Query accuracy} We test the query accuracy (using MRE and MAE) of different methods. The results are shown in \figref{fig:exp_MRE_MAE_eps} and \figref{fig:exp_MRE_MAE_m}. Over different privacy budget, FedGroup consistently outperforms the DP baseline by a large margin (30-50\% improvement). Across different input sizes (the number of federations $m$), FedGroup also provides a significant accuracy improvement.

\begin{figure}[t!]
    \centering
       \begin{subfigure}[b]{0.23\textwidth}
           \centering
           \includegraphics[width=\textwidth]{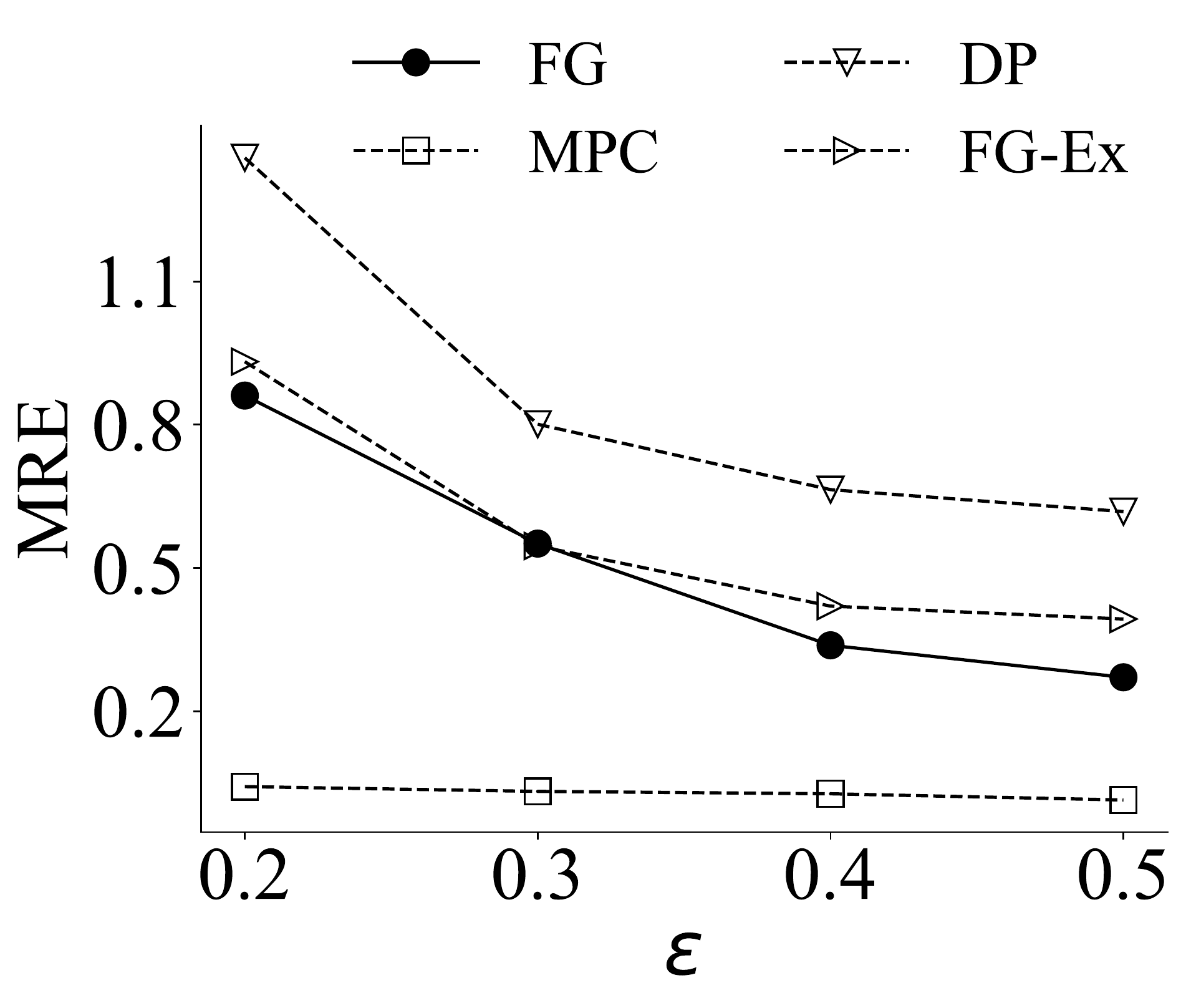}
           \caption{MRE vs. $\epsilon$.}
           \label{subfig:exp_MRE_eps}
       \end{subfigure}
       \hfill
        \begin{subfigure}[b]{0.23\textwidth}
           \centering
           \includegraphics[width=\textwidth]{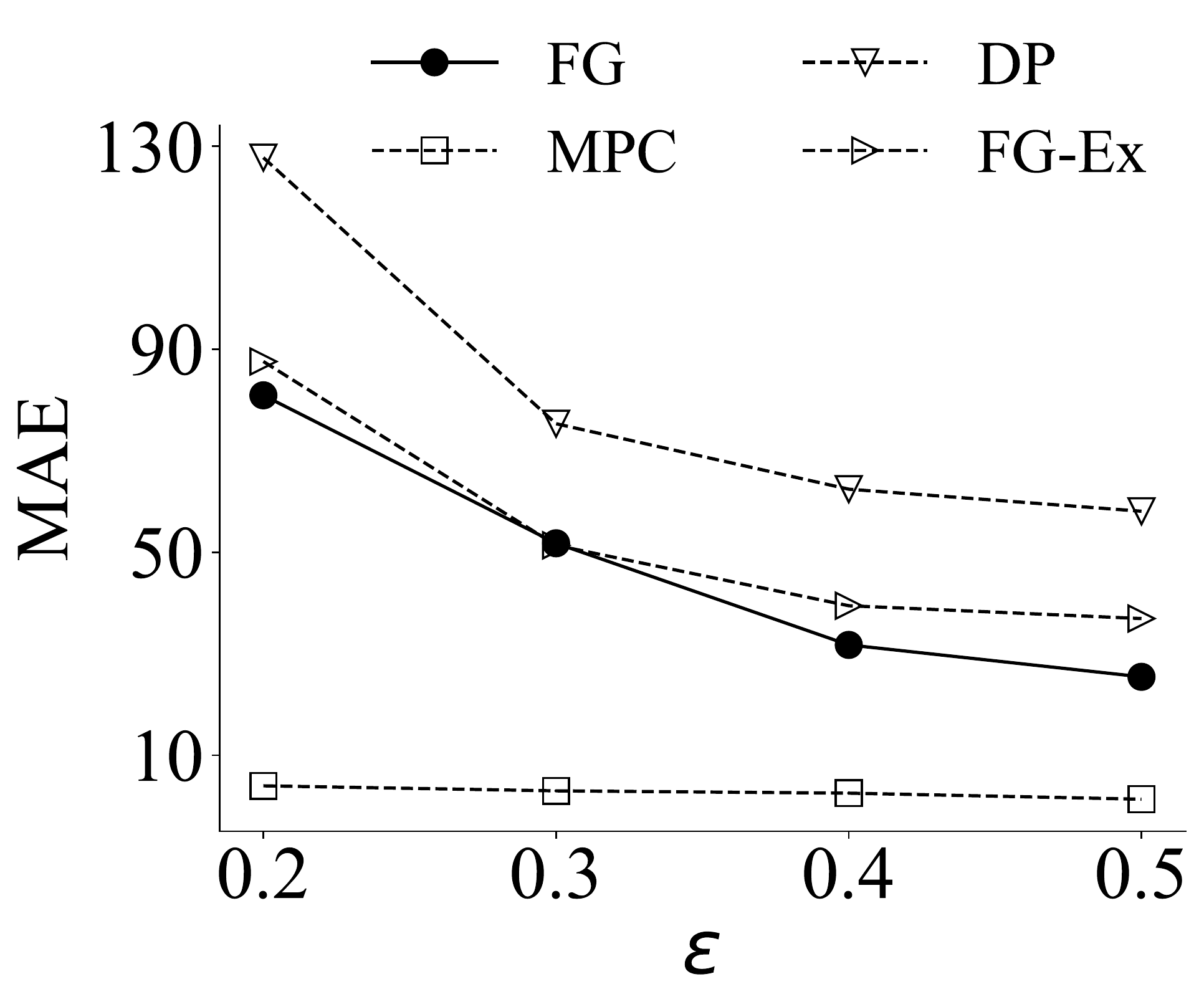}
           \caption{MAE vs. $\epsilon$.}
           \label{subfig:exp_MAE_eps}
       \end{subfigure}
       \hfill
        \begin{subfigure}[b]{0.23\textwidth}
           \centering
           \includegraphics[width=\textwidth]{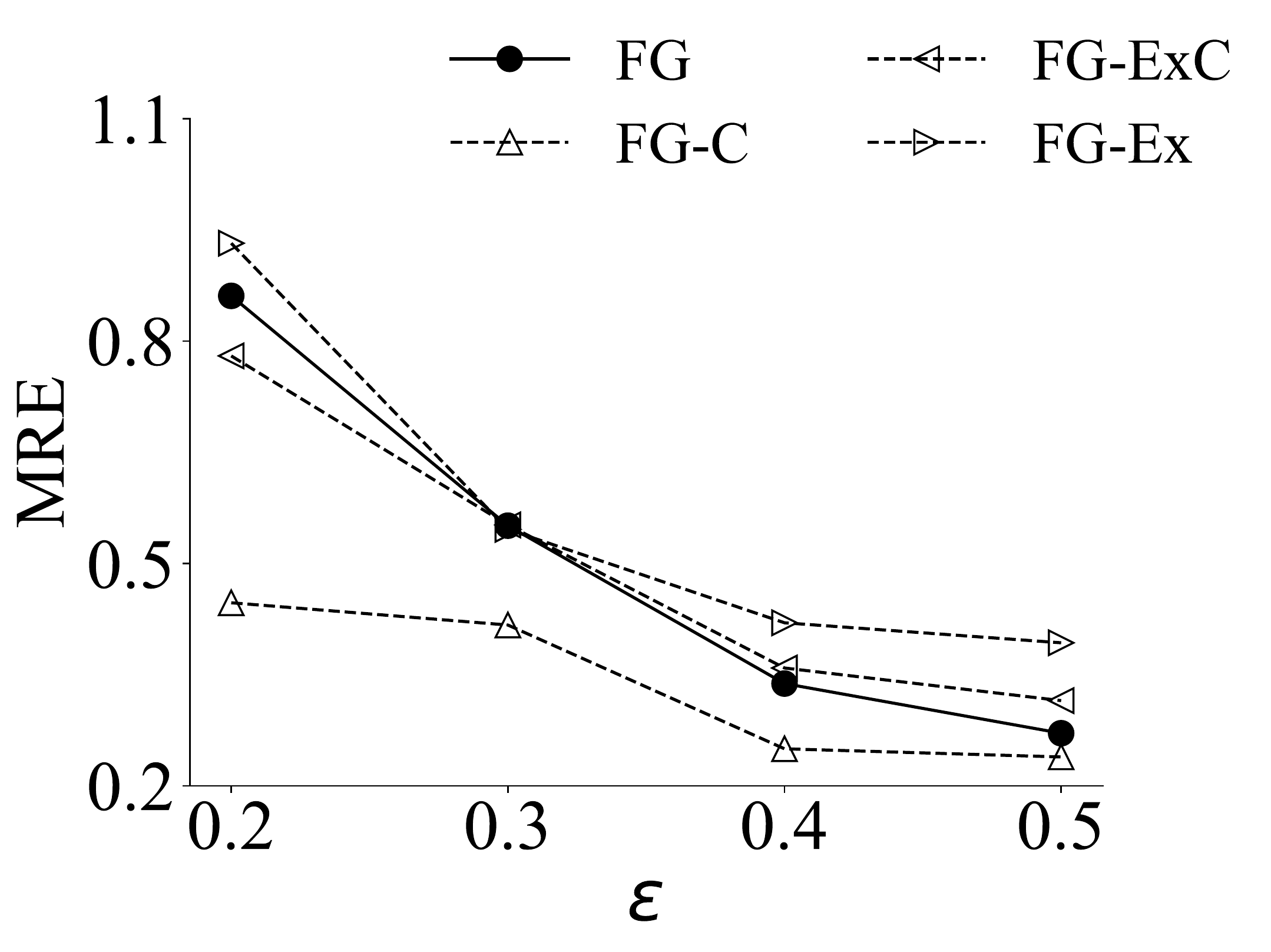}
           \caption{Noisy vs. Clear.}
           \label{subfig:exp_MRE_eps_C_vs_N}
       \end{subfigure}
       \hfill
        \begin{subfigure}[b]{0.23\textwidth}
           \centering
           \includegraphics[width=\textwidth]{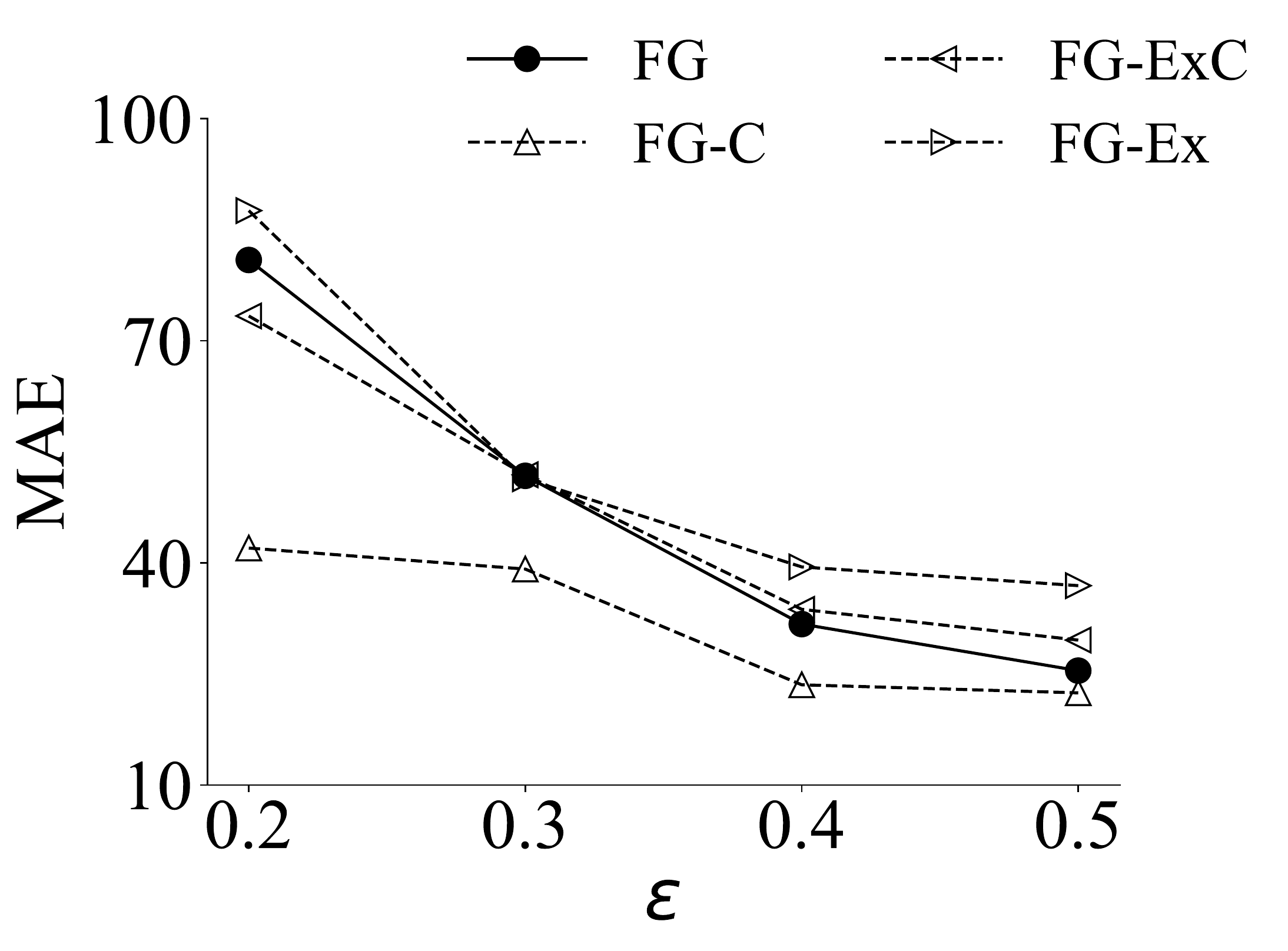}
           \caption{Noisy vs. Clear.}
           \label{subfig:exp_MAE_eps_C_vs_N}
       \end{subfigure}
    \caption{\small End-to-end query accuracy of different methods vs. the privacy budget $\epsilon$. (Gowalla dataset, $ m=500.$) }\label{fig:exp_MRE_MAE_eps}
   \end{figure}

   \begin{figure}[t]
    \centering
       \begin{subfigure}[b]{0.23\textwidth}
           \centering
           \includegraphics[width=\textwidth]{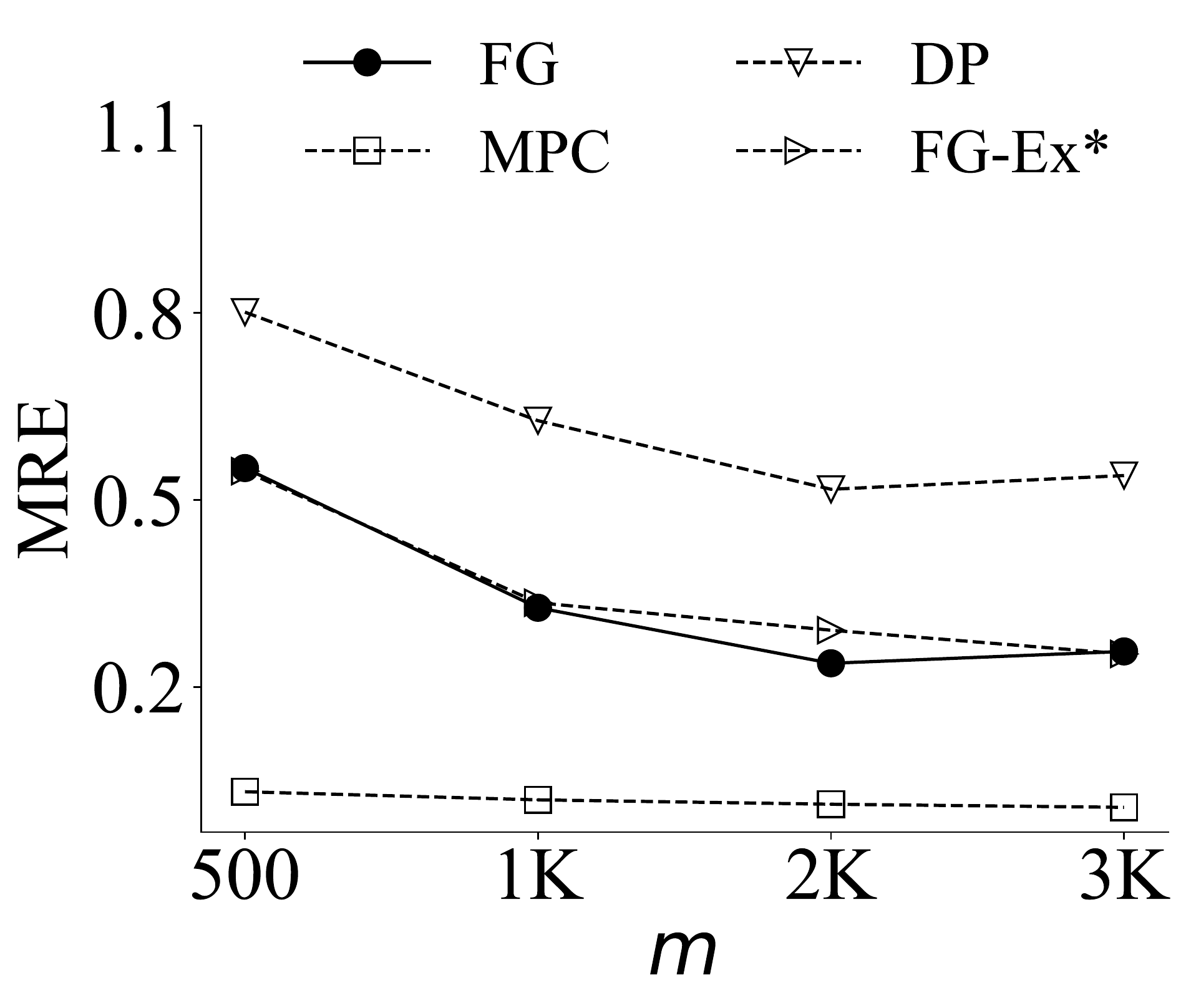}
           \caption{MRE vs. $m$.}
           \label{subfig:exp_MRE_m}
       \end{subfigure}
       \hfill
        \begin{subfigure}[b]{0.23\textwidth}
           \centering
           \includegraphics[width=\textwidth]{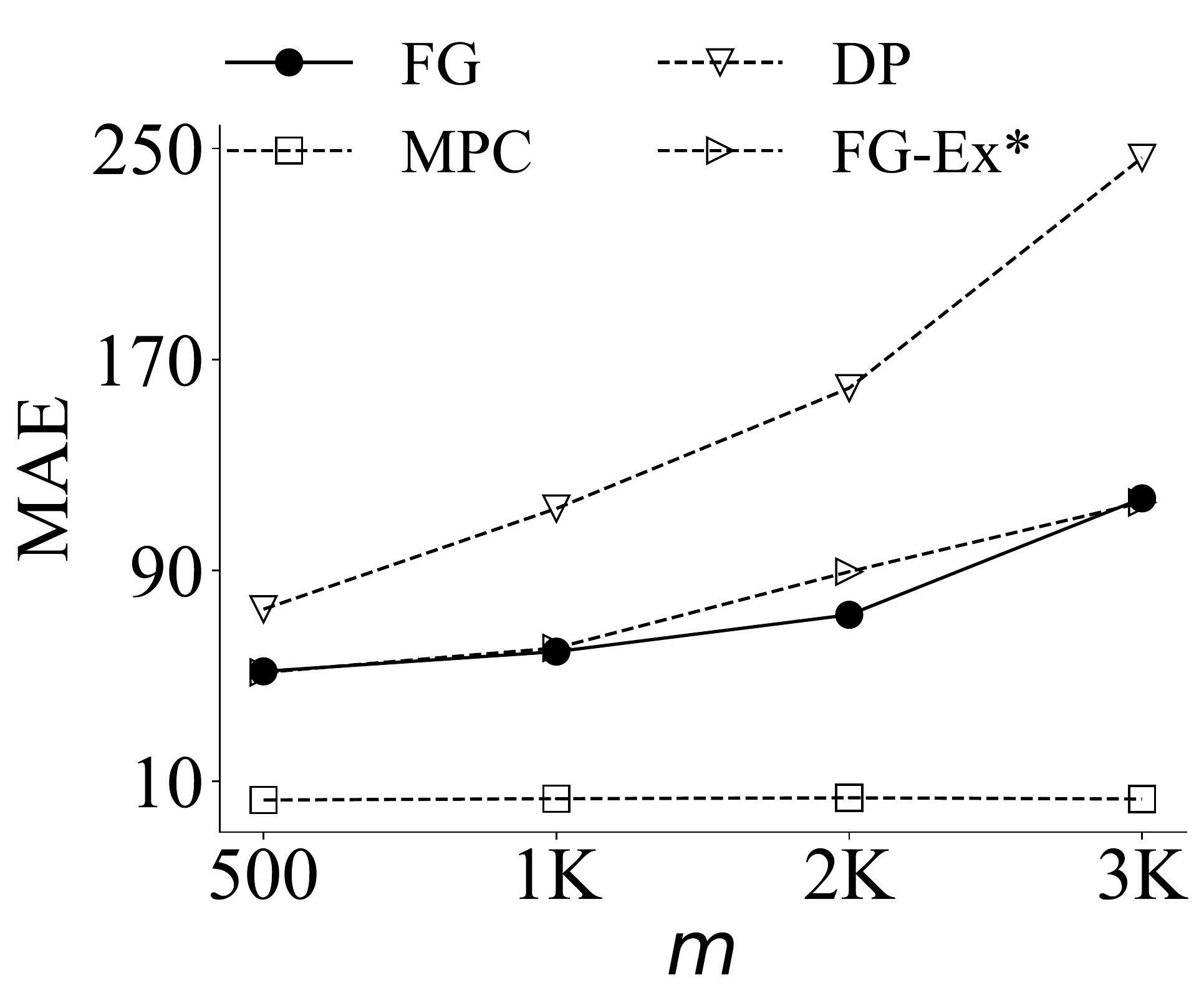}
           \caption{MAE vs. $m$.}
           \label{subfig:exp_MAE_m}
       \end{subfigure}
       \hfill
        \begin{subfigure}[b]{0.23\textwidth}
           \centering
           \includegraphics[width=\textwidth]{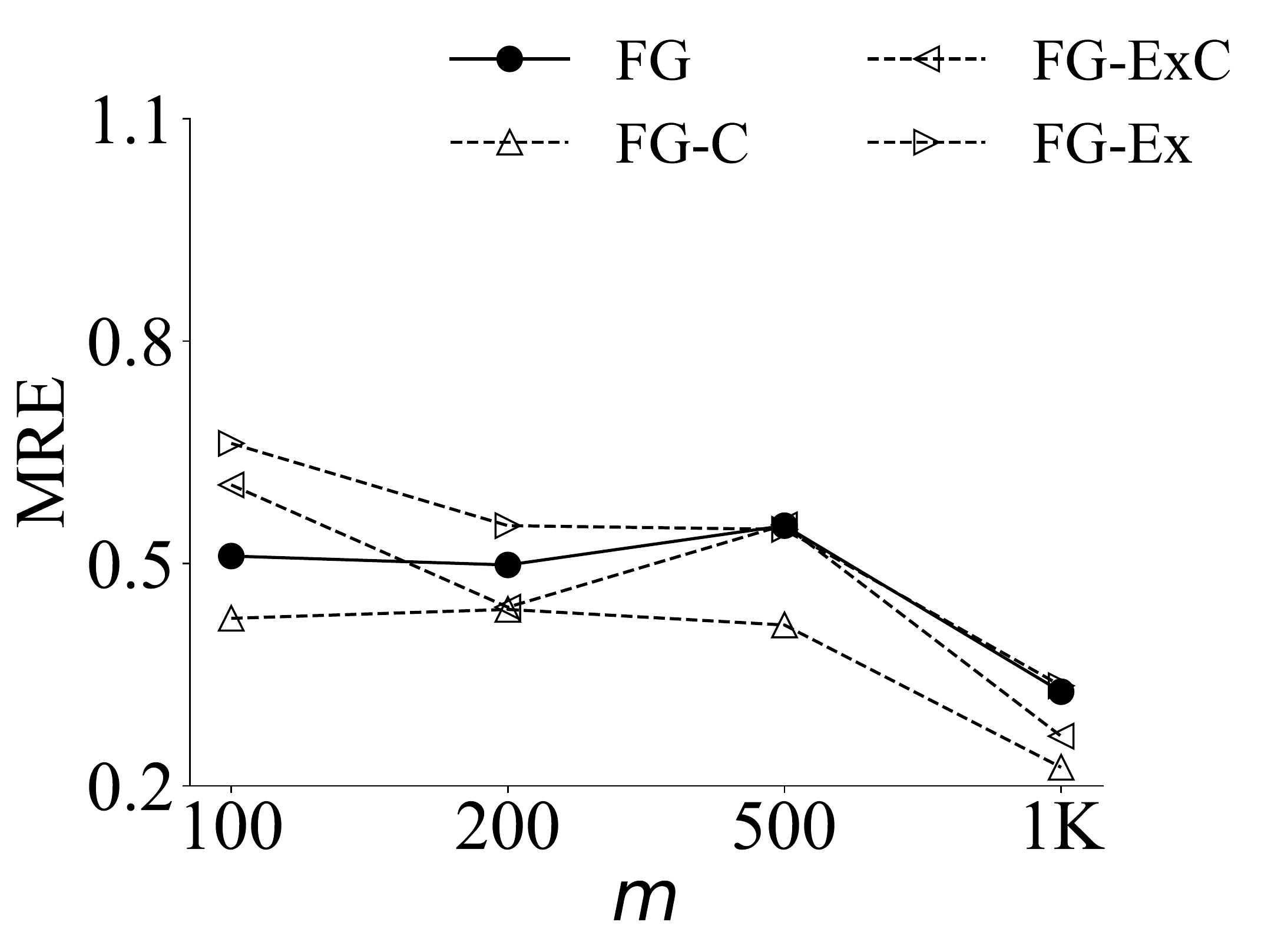}
           \caption{Noisy vs. Clear.}
           \label{subfig:exp_MRE_m_C_vs_N}
       \end{subfigure}
       \hfill
        \begin{subfigure}[b]{0.23\textwidth}
           \centering
           \includegraphics[width=\textwidth]{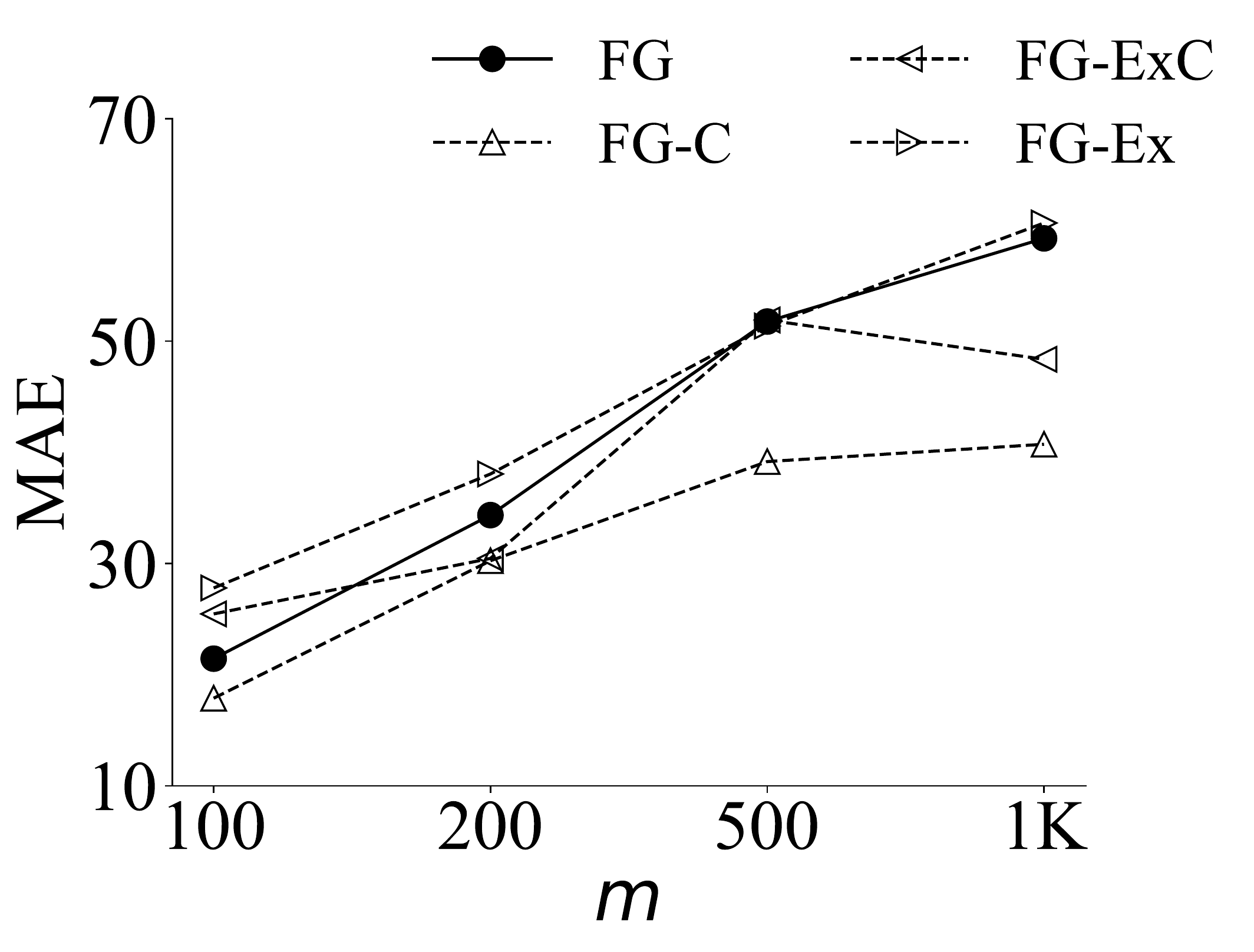}
           \caption{Noisy vs. Clear.}  
           \label{subfig:exp_MAE_m_C_vs_N}
       \end{subfigure}
    \caption{\small End-to-end query accuracy of different methods vs. the number of federations $m$. (Gowalla dataset, $ \epsilon=0.3$.) }\label{fig:exp_MRE_MAE_m}

   \end{figure}

\begin{figure}[t]
    \centering
       \begin{subfigure}[b]{0.23\textwidth}
           \centering
           \includegraphics[width=\textwidth]{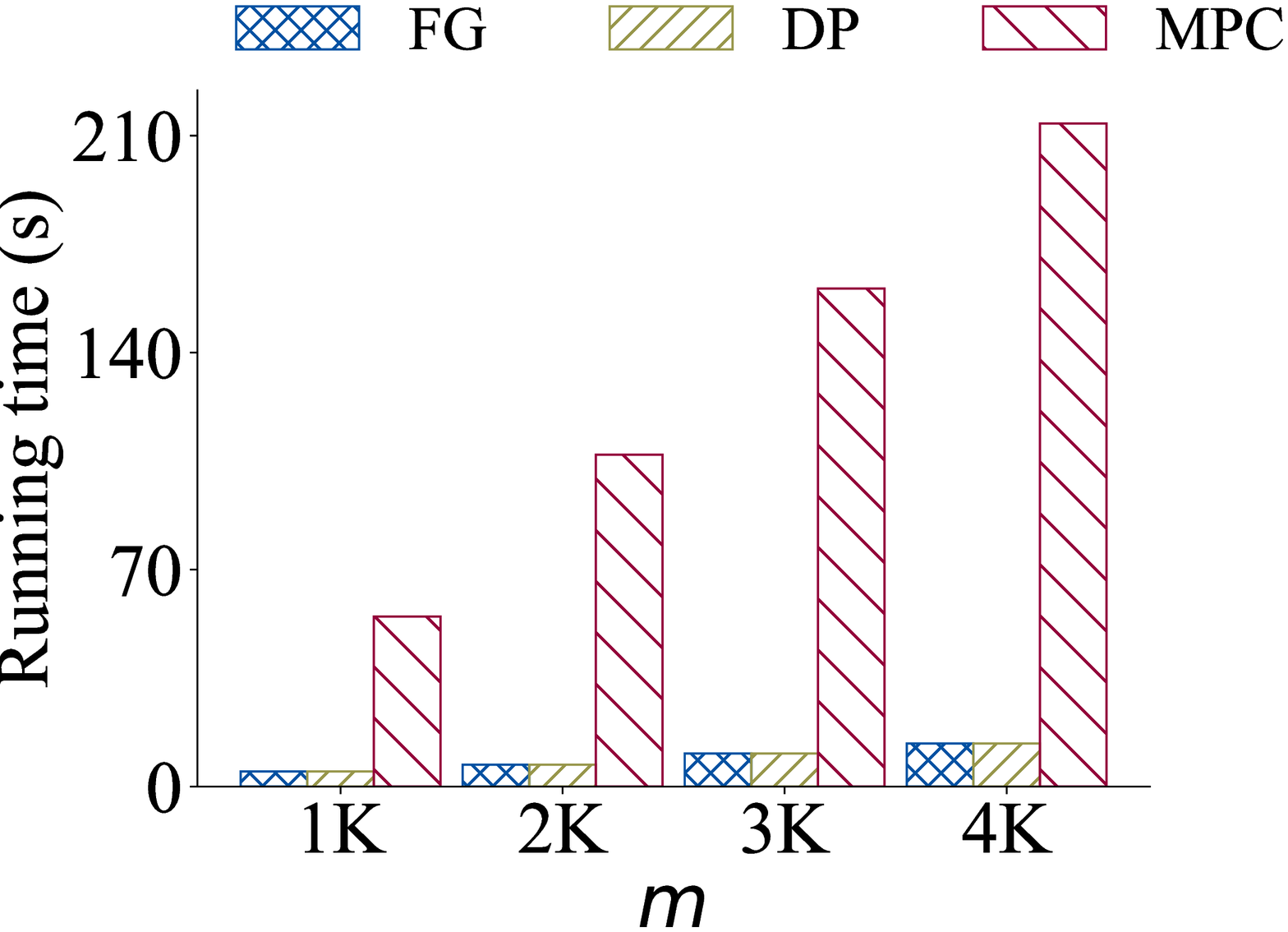}
           \caption{Varying $m$.}
           \label{subfig:exp_time_m}
       \end{subfigure}
       \hfill
        \begin{subfigure}[b]{0.23\textwidth}
           \centering
           \includegraphics[width=\textwidth]{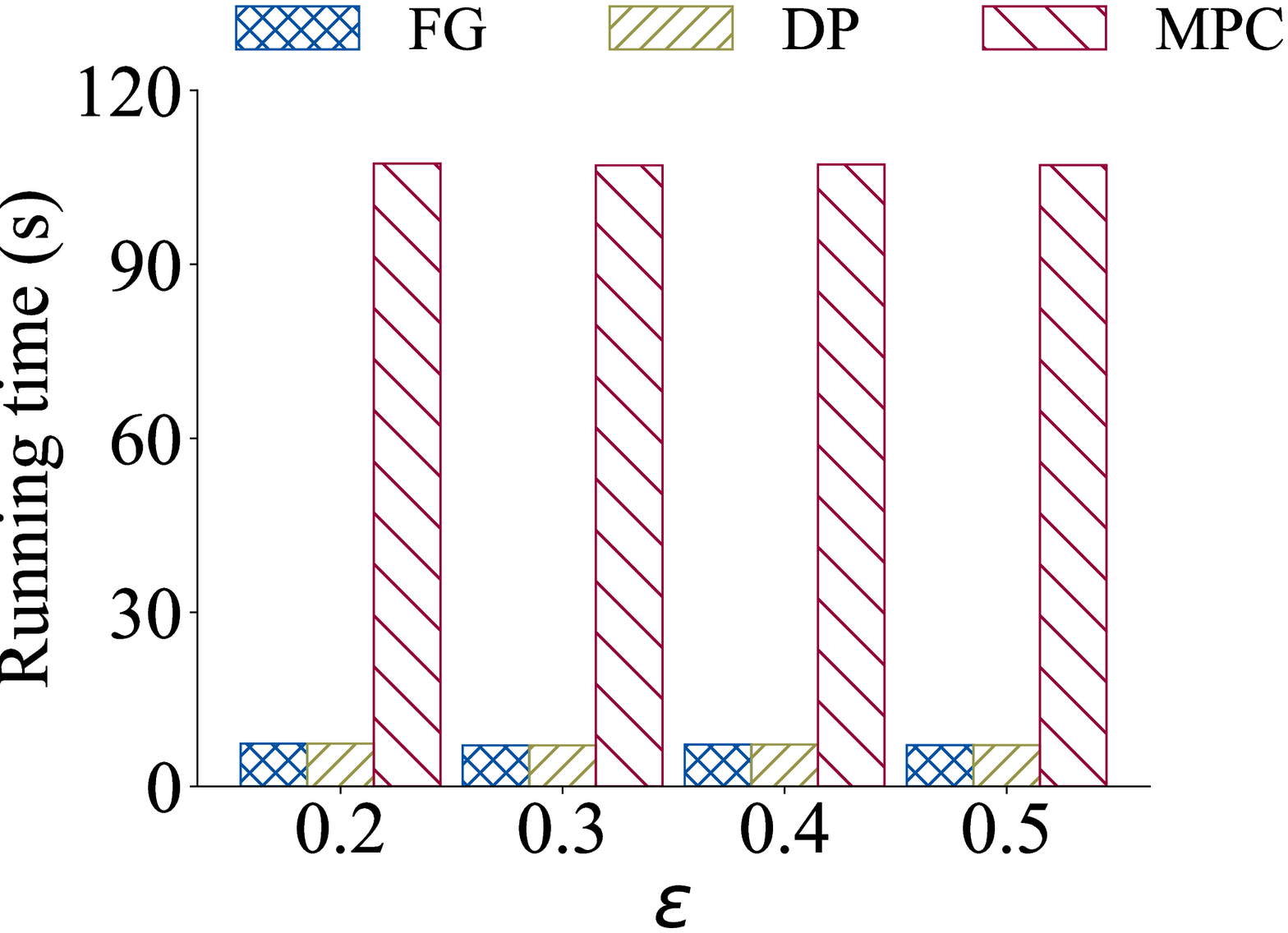}
           \caption{Varying $\epsilon$.}
           \label{subfig:exp_time_eps}
       \end{subfigure}
       \hfill
       \begin{subfigure}[b]{0.23\textwidth}
        \centering
        \includegraphics[width=\textwidth]{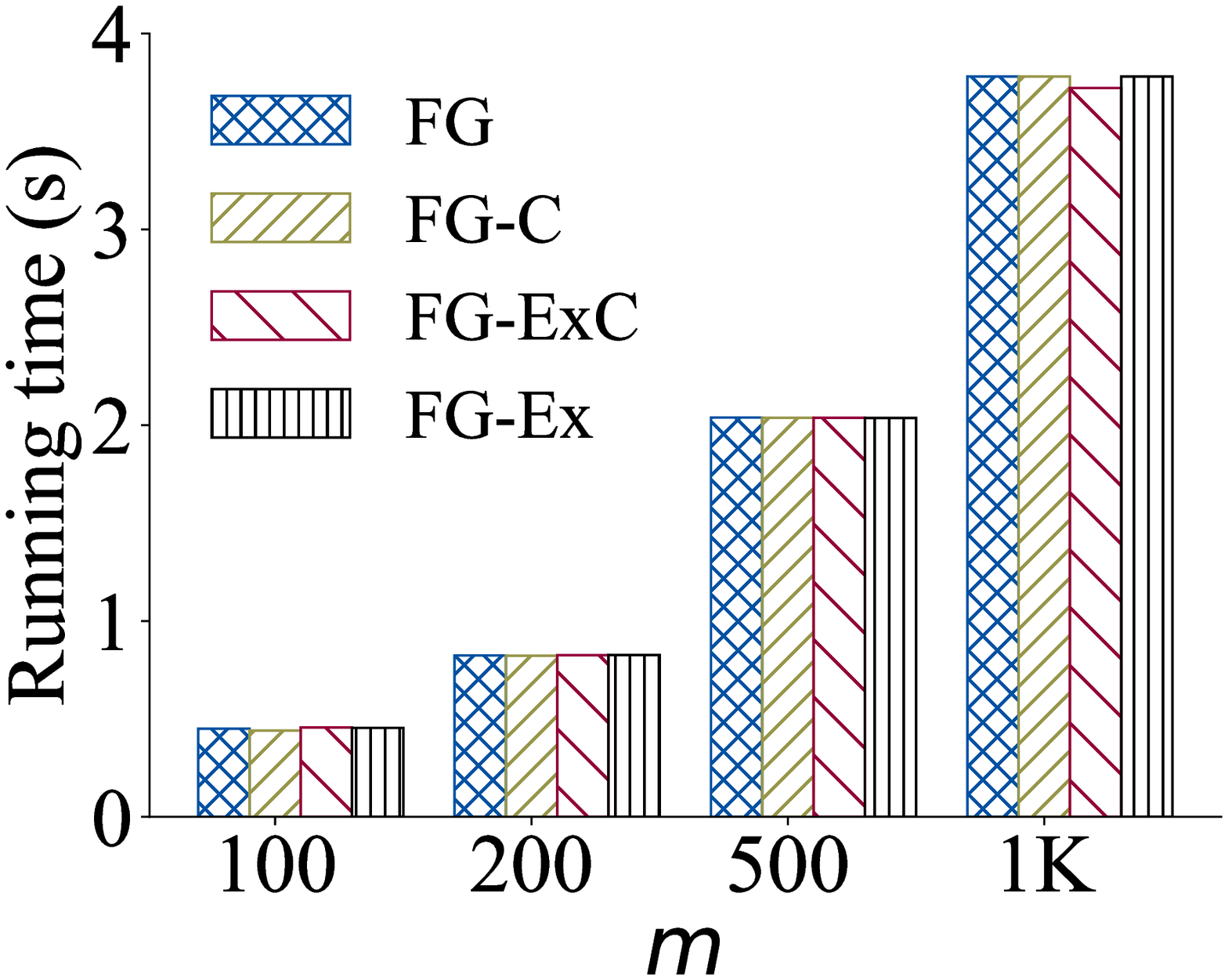}
        \caption{Noisy vs. Clear.}
        \label{subfig:exp_time_m_N_vs_C}
        \end{subfigure}
       \hfill
       \begin{subfigure}[b]{0.23\textwidth}
        \centering
        \includegraphics[width=\textwidth]{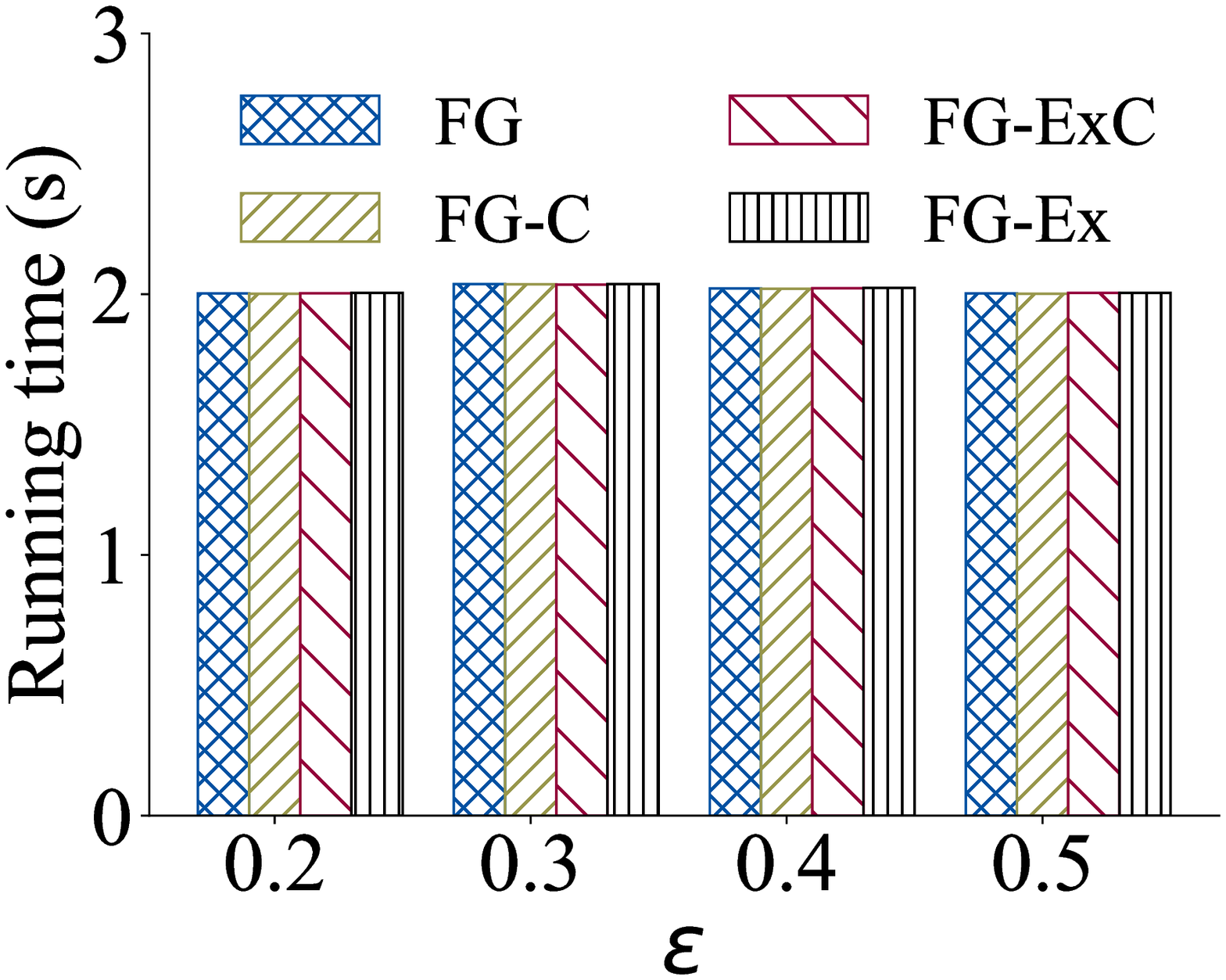}
        \caption{Noisy vs. Clear.}
        \label{subfig:exp_time_eps_N_vs_C}
    \end{subfigure}

    \caption{\small Efficiency (running time) for different control variables. (Gowalla dataset) }\label{fig:exp_time}\vspace{0ex}
   \end{figure}

 \fakeparagraph{Varying $\epsilon$} \figref{fig:exp_MRE_MAE_eps} shows the query accuracy results when we test the compared methods over different privacy budget $\epsilon$ for a spatial federation of $m=500$ data owners using the Gowalla dataset. FedGroup (FG) offers a clear improvement over the DP baseline. Specifically, as shown in \figref{subfig:exp_MRE_eps}, the MRE for the DP baseline at $\epsilon=0.3$ is 0.801 (80.1\%), while FG reduces the error to 0.551 (55.1\%), which provides a $(0.801-0.551)/0.801 * 100\%=31.2\%$ improvement. 
 
 When $\epsilon$ gets larger, indicating a more relaxed privacy requirement, the query accuracy generally improves for all methods, reflected by a smaller MRE and MAE. At $\epsilon=0.5$ (as shown in \figref{subfig:exp_MRE_eps} and \figref{subfig:exp_MAE_eps}), the MRE and MAE for the DP baseline is 0.618 and 58.069, respectively. In comparison, the MRE and MAE of FG improves to 0.271 and 25.448, respectively. The query accuracy improvement of FG over DP is about 56.1\%. The accuracy improvement is more significant than the one of $\epsilon=0.3$ (31.2\%), and the explanation is that when $\epsilon$ gets larger, more accurate spatial similarity graphs are constructed (leading to more correctly identified edges), and FG could assign data owners into smaller number of groups, which eventually leads to a smaller end-to-end query error. 

 It is worth noting that the computationally heavy MPC baseline offers the best query accuracy (due to the $O(1)$ noise) across all $\epsilon$. Despite its impracticability in real-world data sizes, %(see the efficiency results in \secref{subsec:exp_efficiency})
 we include its results here for completeness.

 \begin{figure}[t]
    \centering
       \begin{subfigure}[b]{0.23\textwidth}
        \centering
        \includegraphics[width=\textwidth]{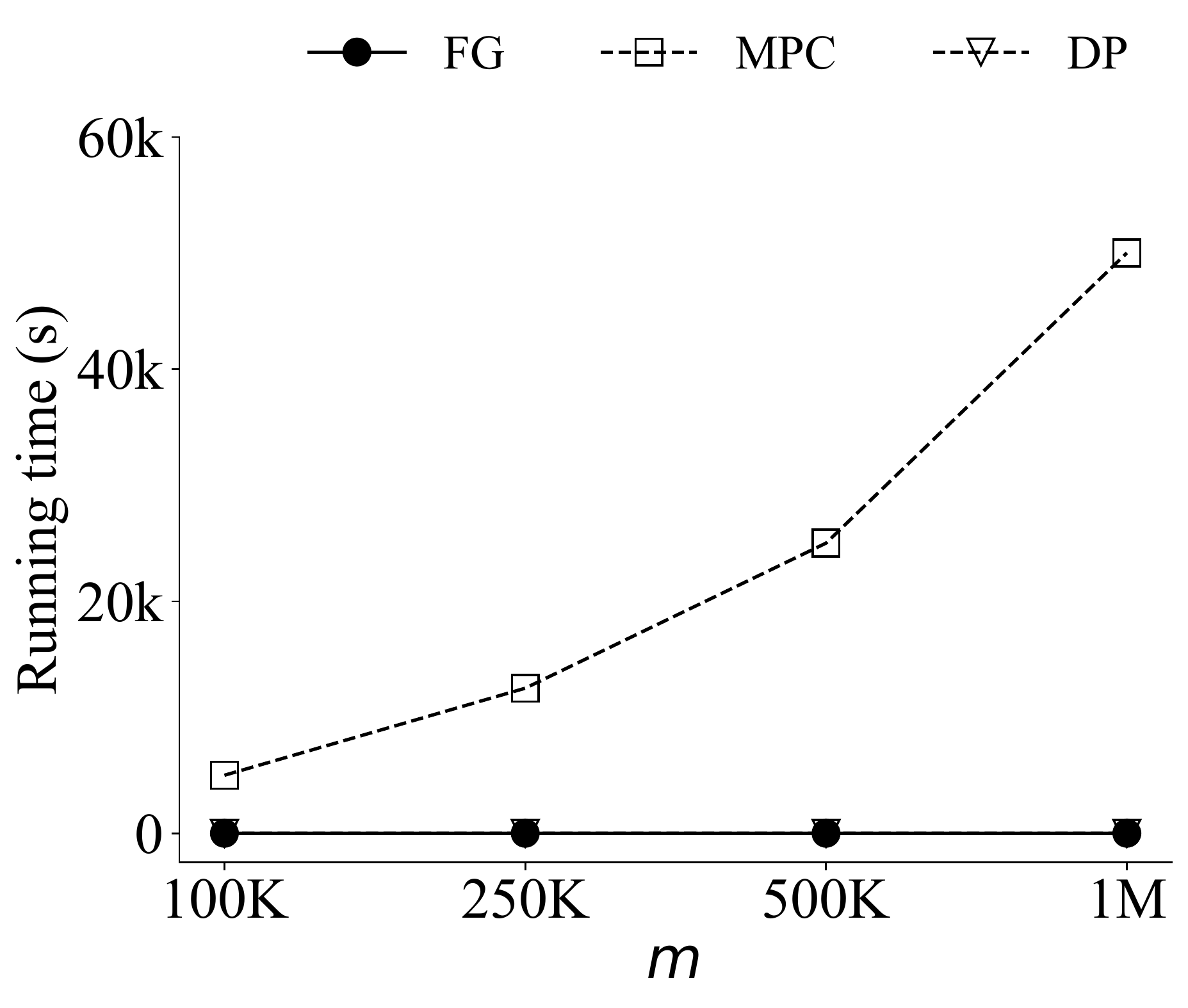}
        \caption{Time vs. $m$.}
        \label{subfig:exp_time_m_scale}
        \end{subfigure}
       \hfill
       \begin{subfigure}[b]{0.23\textwidth}
        \centering
        \includegraphics[width=\textwidth]{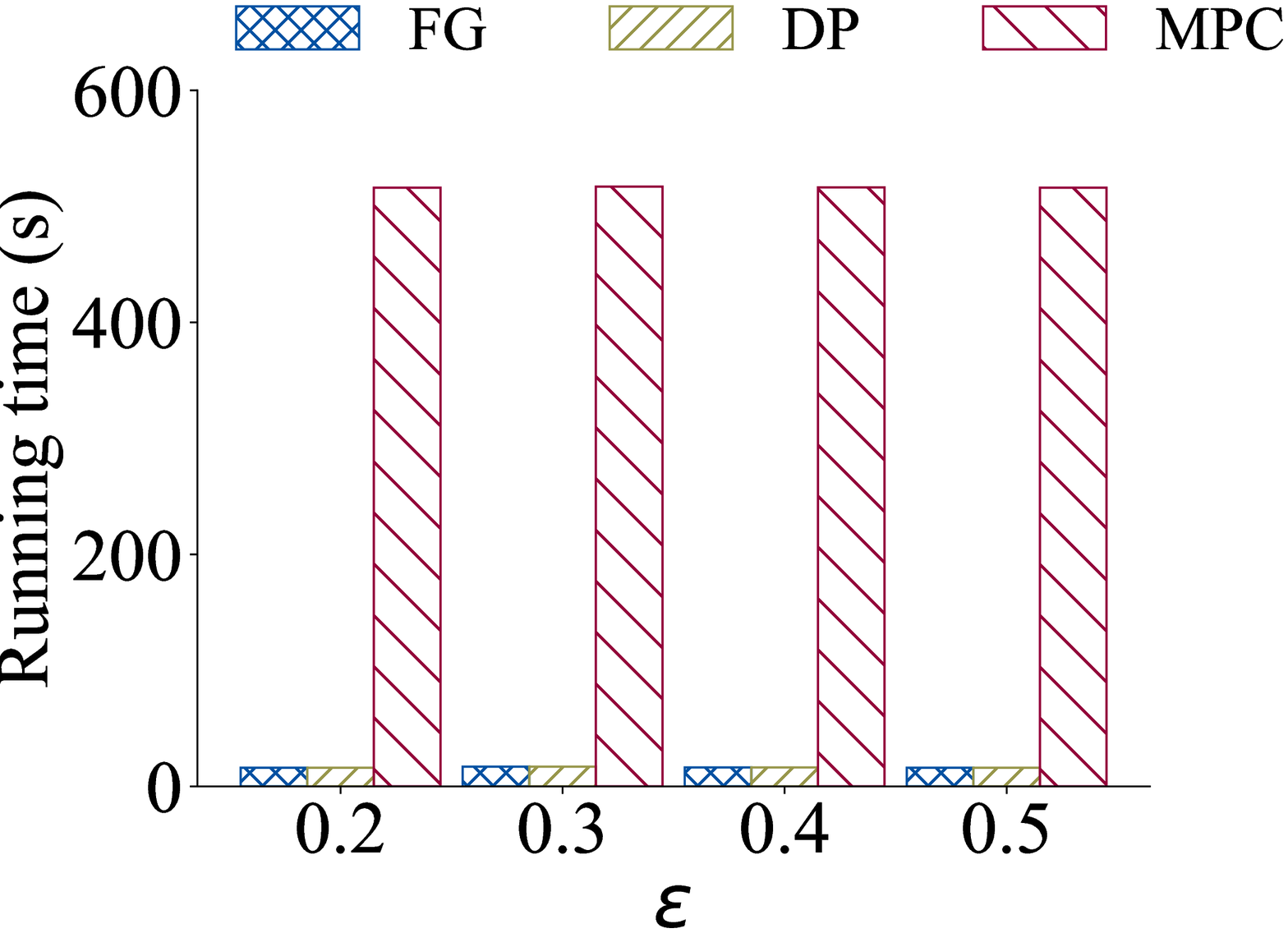}
        \caption{Time vs. $\epsilon$.}
        \label{subfig:exp_time_scale_eps}
    \end{subfigure}
    \begin{subfigure}[b]{0.23\textwidth}
      \centering
      \includegraphics[width=\textwidth]{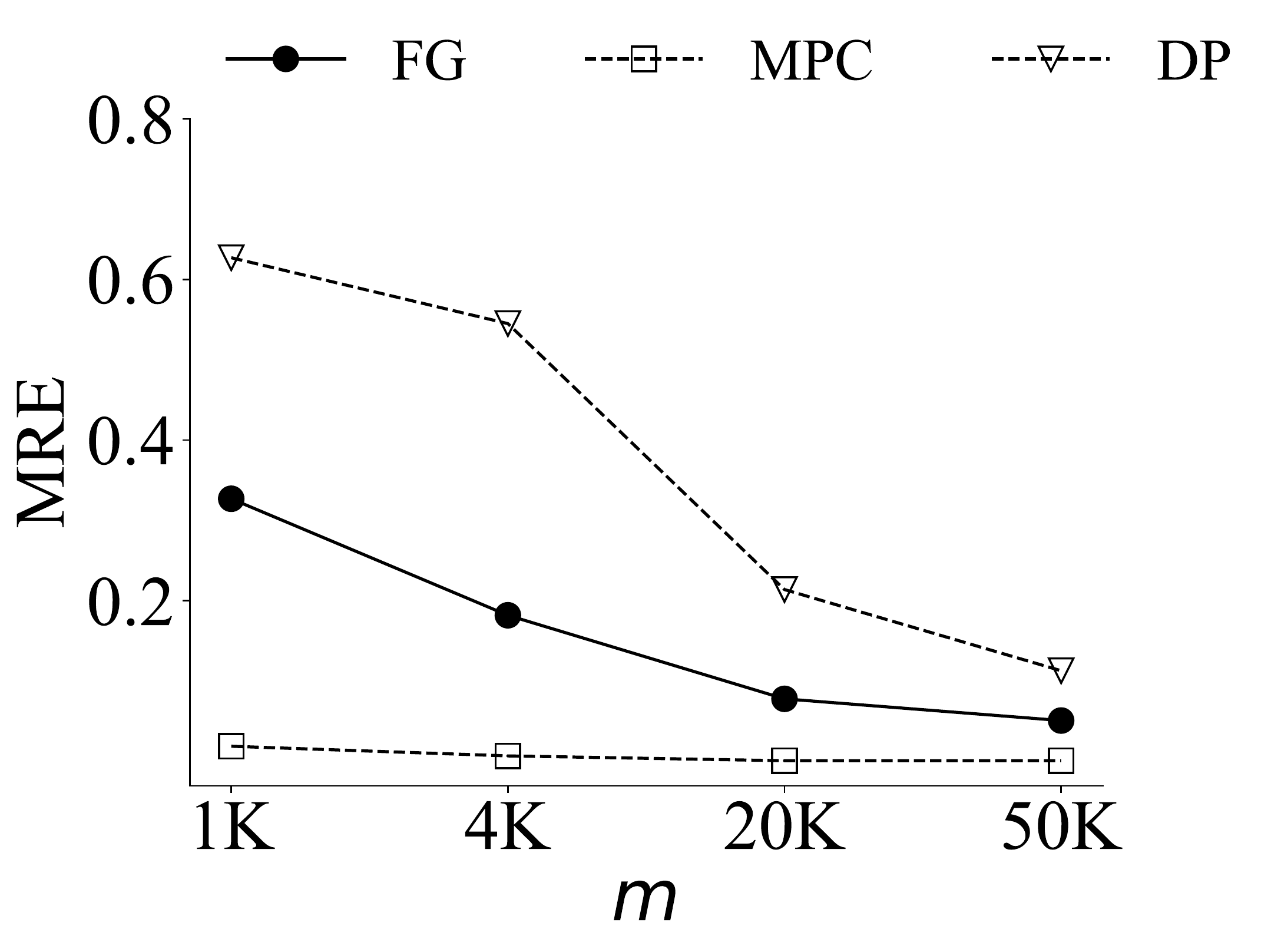}
      \caption{MRE vs. $m$.}
      \label{subfig:exp_scale_MRE}
  \end{subfigure}
  \hfill
   \begin{subfigure}[b]{0.23\textwidth}
      \centering
      \includegraphics[width=\textwidth]{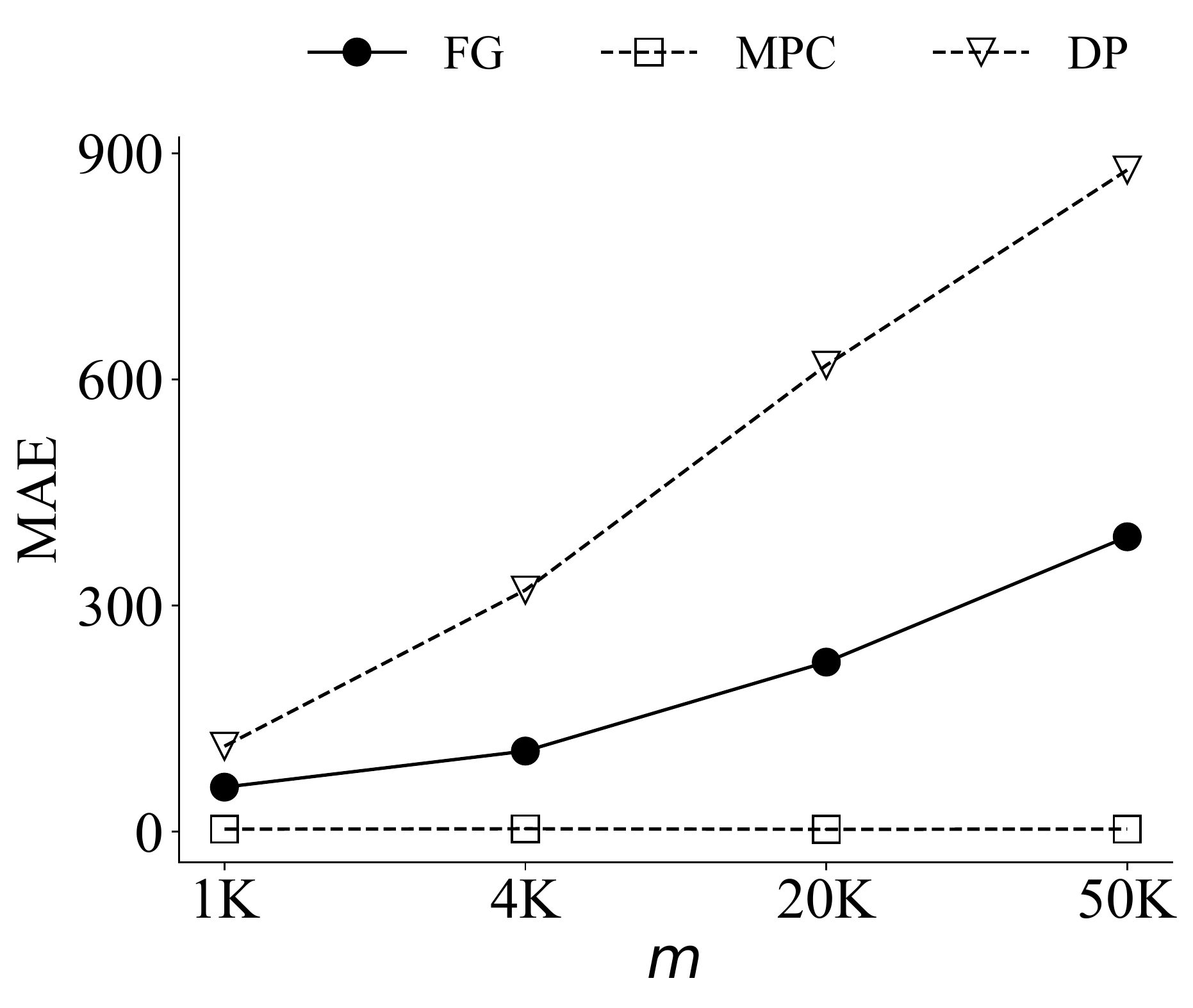}
      \caption{MAE vs. $m$.}
      \label{subfig:exp_scale_MAE}
  \end{subfigure}
  \hfill
    \caption{\small Scalability tests.  (Synthetic dataset)}\label{fig:exp_time_scale}
   \end{figure}
   
 \fakeparagraph{Varying $m$} For a fixed privacy budget $\epsilon$, we test the query accuracy of different methods across different $m$. The results (\figref{fig:exp_MRE_MAE_m}) verify that FG consistently outperforms the DP baseline by a significant margin. As $m$ grows from 500 to 3K, the MAE of the DP baseline grows from 75.3 to 246.5. In comparison, our FG provides a lower MAE, which grows from 51.8 to 117.4. From \figref{subfig:exp_MAE_m}, we could see that the MAE of FG grows more slowly than the DP baseline. 

 When $m$ increases, the true query count also increases. This results a decreasing trend for MRE for all methods, as shown in \figref{subfig:exp_MRE_m}. When $m=2K$, the MRE of the DP baseline is 0.517, while our FG's MRE is 0.238, which provides a 53.97\% improvement. For different $m$, the MRE of DP is about 2$\times$ as large (worse) as the one of our FG. \\

 \fakeparagraph{Query efficiency} We then conduct experiments to test the query efficiency for different methods. The results are shown in \figref{fig:exp_time}. The results on the scalability tests on the synthetic dataset are shown in \figref{fig:exp_time_scale}.

 \fakeparagraph{Gowalla dataset} For smaller $m$ in the range of 1K-4K, we show the query execution time of different methods in \figref{fig:exp_time}, using the Gowalla dataset. The query time using the MPC method grows linearly as $m$ grows (\figref{subfig:exp_time_m}), and stays at the same level across different $\epsilon$ (\figref{subfig:exp_time_eps}), but it is orders-of-magnitude larger (slower) than FG and DP. 

 \figref{subfig:exp_time_m_N_vs_C} and \figref{subfig:exp_time_eps_N_vs_C} show that the running time is not impacted much by whether we use the ground-truth graph or the noisy graph. The time grows roughly linearly when $m$ increases, but all FG variants finish within 4 seconds when $m=1K$, which is not comparable to the time-consuming MPC baseline.

 \fakeparagraph{Scalability tests} We further verify the scalability of our FG method by testing on the synthetically generated dataset. \figref{subfig:exp_time_m_scale} shows the query execution time when we scale $m$ to a million. Despite the linear growth of running time of MPC, it requires more than 13 hours to execute a query when $m=1M$. This is not acceptable in real-world LBS applications. In comparison, both the DP and our FG answers the query within 20 seconds at the largest input size. 
  
 For a given $m$, the running time of different methods do not vary much over different $\epsilon$ (\figref{subfig:exp_time_scale_eps}). In addition to the efficiency results, we also include some query accuracy results when we test the methods on extremely large inputs from the synthetic dataset (shown in \figref{subfig:exp_scale_MRE}-\figref{subfig:exp_scale_MAE}). These results further verify that our FG method is consistently outperforming the equally efficient DP method. \\

 \fakeparagraph{Summary} We conduct extensive experiments to examine both the effectiveness (in terms of query accuracy) and the efficiency of our proposed method FedGroup. In terms of query accuracy, FedGroup provides up to 50\% improvement over the DP baseline. Though the computationally heavy MPC baseline delivers the best accuracy, its inefficiency prevents it from real-world application. The FedGroup runs orders-of-magnitude faster than the MPC baseline and executes a query within 20 secondes for a million data owners. 
 
\section{Related Work}
\label{sec:relatedWork}

In this section, we first review the works on general-purposed data federations, and then discuss specific works on spatial data federation, as well as other spatial data management techniques considering data privacy. 

\fakeparagraph{General-purposed data federations} 
The concept of data federation dated back to the 90s when the need arises to manage autonomous but cooperating database systems \cite{sheth1990federated}. In recent years, there are a growing number of works on optimizing and adopting secure multiparty computations to build privacy-preserving data federations \cite{DBLP:journals/pvldb/BaterEEGKR17,bater2018shrinkwrap,bater2020saqe}. Bater et al. first implemented the MPC-based data federation system called SMCQL \cite{DBLP:journals/pvldb/BaterEEGKR17}, and then improved the system to Shrinkwrap \cite{bater2018shrinkwrap} by using differential privacy to reduce the amount of dummy records inserted during query execution in order to make the computation \textit{oblivious}. Its later effort, SAQE \cite{bater2020saqe},  further improves the tradeoff among utility, efficiency, and privacy by considering approximate query processing, as a by-product of injecting necessary noise for achieving differential privacy. Conclave \cite{DBLP:conf/eurosys/VolgushevSGVLB19} is another system which adopts secure query processing, \ie the private data are \textit{confidential} during the execution of the queries. 

These security-based systems offer strong privacy guarantee, however the efficiency issue is a bottle-neck. Most of the works are limited to only two parties in the secure computation, and experimental results \cite{tong2022hu} show that the current system is still far from being scalable in real-world data sizes. For example, joining a private table L (shared by multiple parties) and a public table R takes more than 25 minutes, even if R contains only 100 objects.

\fakeparagraph{Specialized spatial data federations} In an effort to develop specialized and more optimized spatial data federation, Shi et al. \cite{shi2021efficient} first studied how to efficiently perform approximate range aggregation queries without considering privacy. Later, a MPC-based system HuFu \cite{tong2022hu} was built, with specialized optimization made for spatial operators. Though significant improvements are made over other existing works (\eg Conclave), the running time is still a major concern. In our experiments, we demonstrate even the most simplified MPC baseline runs orders-of-magnitude slower than our proposed solution.

\fakeparagraph{Privacy-preserving spatial data management} In parallel, there are a fruitful amount works of spatial data management using DP or local DP (LDP). The DP model masks the existence of a single location record \cite{cormode2012differentially} 
or an entire trajectory \cite{he2015dpt}. There are also other efforts on protecting location data \cite{xiao2015protecting,DBLP:conf/icde/TaoTZSC020} or using LDP and its variants on offering theoretical guarantee to protect the location data \cite{chen2016private,andres2013geo,DBLP:journals/pvldb/CunninghamCFS21}. 
The DP and LDP model both differ from our privacy model, where a number of data owners exist.

\section{Conclusion}
\label{sec:conclusion}

In this paper, we target at the Federated Privacy-preserving Range Counting (FPRC) problem, where the number of federations (data owners) is large. By utilizing the social ties between data owners (such as family members), we propose a grouping-based framework called FedGroup. It reduces the amount of noise required by only injecting one instance of Laplace noise per group, while the DP baseline requires one instance of noise for each data owner, resulting an unacceptably large scale of noise, overweighing the true query answer. In addition, FedGroup is vastly more efficient than the MPC-based baseline, especially when we scale to a million data owners. Extensive experimental results verify the accuracy improvement and efficiency advantage of our proposed solution. For future works, we look into more advanced ways to compress data \cite{DBLP:journals/pvldb/LiuZSC20,DBLP:conf/sigmod/LiuS022,DBLP:conf/icde/LiuSC21} to achieve better privacy/utility tradeoff.

\bibliographystyle{plain}
\bibliography{lion}

\newpage
\appendix

\section{Additional technical details}
\label{sec:appendix}

For ease of presentation, the major notations used in this paper are summarized in Table \ref{tab:table1}. 

\begin{table}
	\centering \vspace{0ex}
	{\small\scriptsize
		\caption{\small Major notations used in this paper. \label{tab:table1}}
		\begin{tabular}{l|m{6cm}}
			\hline
			{\bf Symbol} & {\bf \qquad \qquad \qquad\qquad\qquad Description} \\ \hline 
			$m$   & The number of data owners (silos) \\
			\hline
			$u_1, \ldots, u_m$   & The data owners  \\
			\hline
			$D_1, \ldots, D_m$ & The database $D_i$ for data owner $i$ (\aka data silos) \\
			\hline
			$D$ & $D= \cup_{i=1}^m D_i$ is the spatial data federation.  \\
			\hline
			$\epsilon$  & The privacy budget for the query\\
			\hline
			$G_s=(V, E)$ & The spatial similarity graph\\
			\hline
			$w, w(e), \tilde{w}$ & The weight $w$ of the edge $e$, and the noisy weight $\tilde{w}$ \\
			\hline
			$T, |T|$ & The grid structure and the number of  grids \\
			\hline
			$v_i, \tilde{v_i}$ & The counts vector $v_i$ for a data owner $u_i$ and its noisy version. \\
			\hline
			$r$ & The threshold for determining whether two data owners are similar (using their spatial similarity of their databases) \\ 
			\hline
			$g$ & A $r$-group \\
			\hline
			$g_1, g_2, \ldots, g_\lambda$ & The assigned groups for data owners \\
			\hline
			$\lambda$ & The number of groups \\ 
			\hline
			$d$ & The largest degree of the graph \\
			\hline
			
		\end{tabular}
	}\vspace{0ex}
\end{table}

\subsection{Problem definition}

For easier understanding, we use a toy example in \figref{fig:problem} to illustrate the above concepts. 

\begin{example} 
	\label{example:range_count}
	As shown in \figref{fig:problem}, there are $m$ data owners (denoted by the mobile device icons), each owning a database $D_i$. Collectively, the union of all data silos forms the spatial data federation $D= \cup_{i=1}^m D_i$. Each data silo $D_i$ has a spatial database \cite{tong2022hu,shi2021efficient,DBLP:conf/kdd/DuTZTZ18}, denoted by the black dots on the grids. The range counting query $Q$ in this example is denoted by the dotted red circle (we omit the center point $q_0$ and the range $r$ in the figure). The red dots inside the circle are the data records falling within the range of the query $Q$. For example, for the first owner $u_1$'s database $D_1$, the count $Q(D_1)=3$. For the $i$-th data owner $u_i$'s database $D_i$, the count $Q(D_i)=2$. The range counting query in this example should return: 
	
	\begin{equation}
		Q=\sum_{i=1}^mQ(D_i). \label{eq:true_answer}
	\end{equation}

	It is the summation of the partial range count $Q(D_i)$ for each data silo $D_i$. 

\end{example}

\subsubsection{Data privacy and compute privacy}

In our problem, we consider both the input data privacy (\aka \textit{data privacy}) and the oblivious query evaluation (\aka \textit{compute privacy}) \cite{he2021practical} in the spatial data federation. 

For data privacy, we first identify the private data that needs to be protected. In our problem, the private data concerned is the spatial database $D_i$ for each data owner $u_i$. 

As for compute privacy, we refer to privacy leakage in the computation process. It refers to the instruction traces, program counters, and the execution time of a function based on the private inputs, which refer to $D_i$ for each data owner $u_i$.

\subsubsection{Adversary model}

For (1) the end user: an end user issues a range counting query $Q(r, q_0)$ to the service provider, and receives the query result $\tilde{Q}$. The end user is \textit{curious} and tries to obtain as much private information as possible from the query result $\tilde{Q}$. The private information here refers to the private database $D_i$ owned by each data owner $u_i$. 

For (2) the service provider $S$: upon receiving the query $Q(\cdot)$ from the end user, the service provider coordinates and executes the query and obtains an answer $\tilde{Q}$ to be returned the end user. During the execution of the query, the service provider may interact with data owners $u_1, \ldots, u_m$ and receives intermediate results. For example, in \figref{fig:problem}, $S$ receives the partial results 3 from $D_1$, $Q(D_i)$ from $D_i$, and 2 from $D_m$, respectively. $S$ is also assumed to be semi-honest, so $S$ tries to obtain as much private information about $D_1, \ldots, D_m$ as possible. In addition, $S$ also tries to obtain privacy disclosure \textit{during} the computation process, which refers to the evaluation process, including the program counters, the execution time and the instruction traces during its coordination of query execution together with the data owners. 

For (3) the data owners: each data owner $u_i$ is also assumed to be semi-honest and is \textit{curious} about the private information from other data owners. For example, $u_1$ may be curious about $u_3$ and tries to gain private information about $D_3$ during the execution of the query. The leakage may happen \textit{during} the computation process when $S$ coordinates each data owner to jointly compute the results.

\subsection{Baselines}
\label{subsec:baseline}

We start with a case study showing why \figref{fig:problem} fails the privacy requirements in \defref{def:problem}. Then we show two straw-man baselines using DP and MPC. 

\subsubsection{A case study of privacy breakdown}
\label{subsubsec:case_study}

It is not hard to see that the computation process shown in \figref{fig:problem} fails to satisfy the privacy requirements in FPRC as defined \defref{def:problem}, as the true answer computed by each party is directly shared with the semi-honest service provider $S$. In fact, R1 and R2 are both violated, and R3 is satisfied because there is no multiparty computation involved. We defer the detailed analysis to the appendix. 

\subsubsection{The DP baseline}
\label{subsubsec:dp_baseline}

We introduce a DP straw-man baseline by applying the Laplace mechanism to satisfy the privacy requirements, and then we analyze the scale of error introduced. 

\fakeparagraph{For each data owner} Upon receiving the query $Q(r, q_0)$, each data owner $u_i$ computes the partial true result by running the query on its private database to obtain $Q(D_i)$. Then, a single instance of Laplace noise injected to obtain:

\begin{equation}
	\tilde{Q}(D_i) = Q(D_i) + \text{Lap}(1/\epsilon). \label{eq:partial_sum}
\end{equation}
 
 The noisy partial result is sent to the service provider $S$. 

\fakeparagraph{Service provider $S$} After receiving the noisy partial count $\tilde{Q}(D_i)$ from data owners, $S$ aggregates all partial counts by taking a summation to obtain:

\begin{equation}
	\tilde{Q} = \sum_{i=1}^m \tilde{Q}(D_i). \label{eq:aggregate}
\end{equation}

Finally, $S$ returns the value $\tilde{Q}$ back to the end user. 

\fakeparagraph{Privacy analysis} It is rather straightforward to see that R1-R3 are satisfied as the sensitivity of the count query for each partial count $Q(D_i)$ is 1 for all data owners.

\fakeparagraph{Utility analysis} First, we notice that the randomness of $\tilde{Q}$ depends on the Laplace noise randomly generated in \equref{eq:partial_sum}. Since each random variable of Laplace noise has mean 0, the summation of $m$ such random variables also has mean 0. Thus, the returned noisy answer $\tilde{Q}$ in \equref{eq:partial_sum} is an unbiased estimator of the true answer $Q$. % given in \equref{eq:true_answer}. 

Thus, to measure the utility, we define the error as the variance of the returned noisy answer $\tilde{Q}$. 

\begin{equation}
	\text{Var}(\tilde{Q}) = \text{Var}(\sum_{i=1}^m\text{Lap}(1/\epsilon)) = m\cdot \text{Var}(\text{Lap}(1/\epsilon)) = 2m / \epsilon^2. 
\end{equation}

From this, we know that the error of $\tilde{Q}$ scales linearly with the number of federations $m$. In our setting, $m$ is potentially large (\eg millions of mobile devices), and we expect that the DP straw-man baseline yields a significant error. 

\subsubsection{The MPC baseline}
\label{subsubsec:smc_baseline}
We then show a straw-man MPC baseline by utilizing MPC techniques for different data owners to jointly compute the query result. The basic idea is that the $m$ data owners, together with the service provider $S$, jointly computes a noisy summation:

\begin{equation}
	\tilde{Q} = \sum_{i=1}^mQ(D_i) + \text{Lap}(1/\epsilon), \label{eq:smc_baseline}
\end{equation}

taking each of their own private data $D_i$ as private inputs in the MPC protocols. 

There is only one instance of Laplace noise injected to the true query result. Then, the service provider $S$ returns the computed result $\tilde{Q}$ to the end user. 

\fakeparagraph{Privacy analysis} The injected Laplace noise in \equref{eq:smc_baseline} ensures that the final returned query result satisfies R1 in \defref{def:problem}. R2 and R3 are both satisfied because that MPC protocol is used, and each data owner's private database $D_i$ is strictly confidential, without disclosing any information to the other parties. 

\fakeparagraph{Utility analysis} The utility of the straw-man SMC baseline is strictly better than the DP baseline, since there is only one single instance of Laplace noise injected. As compared to the $O(m)$ scale of error by the DP baseline, the MPC baseline's error is $O(1)$. 
	
\fakeparagraph{Time complexity} The MPC baseline computes a secure summation of $m$ data owners, plus the one instance of Laplace noise at the end. The time complexity is linear to the number of data owners: $O(m)$, but the communication cost between data owners is significant. Each data owner $u_i$ needs to wait for the partial result from previous $i-1$ data owners to continue the computation process.

\subsubsection{A case study of privacy breakdown}

As \figref{fig:problem} shows, each data owner $u_i$ computes the \text{partial} range count \wrt the data silo $D_i$. The service provider $S$ then aggregates all the partial counts by taking a summation $\tilde{Q} = \sum_{i=1}^m Q(D_i)$. Then, $S$ returns $\tilde{Q}$ back to the end user. This example results fails several privacy requirements in \defref{def:problem}, as illustrated next. 

First, R2 in \defref{def:problem} is definitely violated. The partial count is directly returned to  $S$. Since we assume that $S$ is semi-honest and adversary could exist on $S$, accessing the partial query result $Q(D_i)$ is a direct privacy threat to the private database $D_i$. Changing of one single record in $D_i$ directly impacts the partial count result $Q(D_i)$, and obviously this step does not satisfy the privacy guarantee in R2. 

Next, R1 in \defref{def:problem} is violated, as $\tilde{Q} = \sum_{i=1}^m Q(D_i)$ is directly returned to the end user. Note that $\tilde{Q}$ equals to the true count, as introduced in Example~\ref{example:range_count}. Changing any particular record in any data silo $D_i$ directly influences the query result $\tilde{Q}$ by the sensitivity of 1, and it fails to provide any privacy guarantee. We omit the formal arguments of showing why it does not satisfy $\epsilon$-DP as it is straightforward. 

Last, we note that R3 in \defref{def:problem} is actually satisfied, because the computation is solely computed by each data owner. Strictly speaking, side channel attacks could still be performed by the semi-honest service provider $S$ to observe the query execution time by each data owner in order to infer the size of private database $D_i$. However, since there is no multiparty computation involved in this step, we consider R3 is satisfied. In our adversary model, side channel attacks only occur when there is multiparty computation involved between the service provider $S$ and other data owners.

\subsubsection{DP baseline privacy analysis}

First we show that R2 in \defref{def:problem} is satisfied. It is straightforward as the sensitivity of the count query for the partial count $Q(D_i)$ is 1 for all data silos. Injecting a Laplace noise with the corresponding sensitivity suffices to provide $\epsilon$-DP. 

Then, we show that R1 in \defref{def:problem} is satisfied by applying the parallel composition theorem of DP. If we look at the entire database $D= \cup_{i=1}^m D_i$, which is the union of all data silos, then each data silo $D_i$ is a \textit{disjoint} subset of $D$. Changing one record in any data silo $D_i$ only changes one record in $D$, and vice versa. Thus, the parallel composition theorem of DP applies, \ie the final result computed in \equref{eq:aggregate} satisfies $\epsilon$-DP, because each data owner's partial result (computed by \equref{eq:partial_sum}) satisfies $\epsilon$-DP. 

R3 for the straw-man DP baseline is also satisfied because there is no multiparty computation involved during the entire process of query evaluation. 

\subsection{FedGroup details}

We provide a toy example next to better illustrate the concept. 

\begin{example}[Spatial similarity]
  \label{example:spatial_sim}
  In \figref{fig:example_spatial_sim_graph}, $D_1$ is the spatial database for data owner $u_1$ and $D_2$ is the spatial database for data owner $u_2$. The grid structure $T$ is shown in the figure, which contains $|T|=9$ grids over the spatial domain. Thus, we could construct $D_1$'s count vector as $v_1 = [0\ 1\ 1\ 1\ 0\ 0\ 2\ 5\ 0]$, where each entry corresponds to the count of the a particular grid, starting from the top-left corner and moving down row by row. Similarly, $D_2$'s count vector is $v_2= [0\ 0\ 1\ 1\ 0\ 0\ 1\ 4\ 0]$. Then, the similarity of between $D_1$ and $D_2$ is sim$(D_1, D_2)=\cos(v_1, v_2) = v_1 \cdot v_2 / \|v_1\|\|v_2\| = 24 / 24.658 = 0.973$.  
\end{example}

A toy example of the spatial similarity graph is given next.

\begin{figure}[t!]\centering 
	\scalebox{0.6}[0.6]{\includegraphics{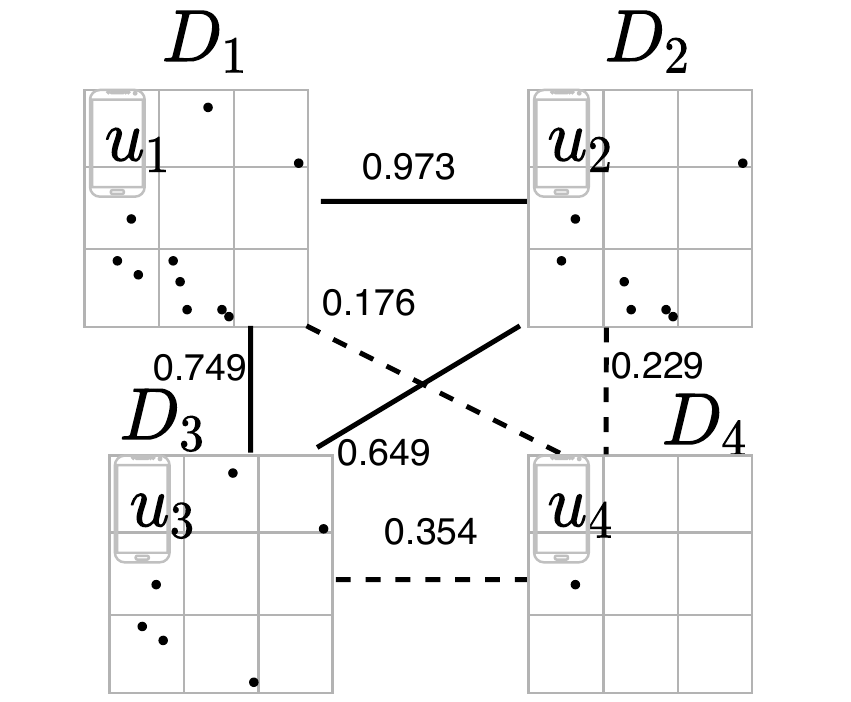}}
	\caption{\small An example of the spatial similarity graph. }
	\label{fig:example_spatial_sim_graph}
\end{figure}

\begin{example}[Spatial similarity graph $G_s$] \figref{fig:example_spatial_sim_graph} gives a toy example of a spatial similarity graph containing 4 data owners. There is an edge between each pair of data owners in $G_s$. The edge's weight is set as the spatial similarity defined in \equref{eq:spatial_sim}. As we have shown in \expref{example:spatial_sim}, the similarity between $u_1$ and $u_2$ equals to $\text{sim}(D_1, D_2)=0.973$, as denoted above the corresponding edge $e=(u_1, u_2)$. Similarly, $w(u_1, u_3)=0.749$, $w(u_3, u_2)=0.649$. Note that these three edges have edge weights larger than the threshold $r=0.5$ (as the threshold $r$ to be used in the $r$-clique in \secref{subsec:grouping}), and are denoted with solid lines in the figure. The other three edges, \ie $(u_1, u_4), (u_3, u_4), (u_2, u_4)$, have edge weights smaller than $r$, and are denoted as dashed lines in the figure. 
\end{example}

\begin{algorithm}[t]
	\DontPrintSemicolon
	\KwIn{$u_1, \ldots, u_m, D_1, \ldots, D_m, r$. }
	\KwOut{$G_s=(V, E)$.}
    $V := \{u_1, \ldots, u_m\}$  \;
    $E := \{\}$ \;
	\ForEach{$u_i \in V$}{
    \ForEach{$u_j \in V$ and $j > i$} {
      $e := (u_i, u_j)$ \;
      $e.\text{weight} := \text{sim}(D_i, D_j)$ \;
      \If{$e.\text{weight} > r$}{
        Insert $e$ to $E$ \;
      }
    }
	}

	\Return{$G_s=(V, E)$}\;
	\caption{\texttt{Non-privately Construct $G_s$} }
	\label{algo:construct_graph_plaintext}
\end{algorithm}

\fakeparagraph{MPC option} The most straightforward way to guarantee all privacy requirements is to use compute edge weight with confidentiality using MPC techniques. Each pair of data owners $u_i$ and $u_j$ uses MPC to securely compute their similarity according to \equref{eq:spatial_sim}, with their private database $D_i,D_j$ and the counts vector $v_i, v_j$ considered as the private inputs to the garbled circuits. Thus, each pair of data owners $u_i$ and $u_j$ uses a 2-party MPC to compute $w=\text{sim}(D_i, D_j)$ with confidentiality, and sends the weight $w$ to the service provider $S$. $S$ sets $s$ as the edge weight for edge $e=(u_i, u_j)$, and inserts $e$ to the edge set $E$. 

\fakeparagraph{DP option} We can also apply standard Laplace mechanism in DP to provide $\epsilon$-DP for each data silos to obtain a noisy version of the counts vector $\tilde{v_i}$ for each data owner $u_i$ (since the sensitivity of the counts within each grid is 1, \wrt to change of one record in a data silo). Then, a noisy similarity score could be computed for each pair of $u_i, u_j$, using $\cos(\tilde{v_i}, \tilde{v_j})$. The noisy weight $\tilde{w}$ is then used as a surrogate score as the spatial similarity score for the private inputs between the two data owners. The service provider $S$ computes noisy score for each pair of data owners, similar to the non-private version. 
However, instead of the \textit{exact} weight computed on the private inputs, now we use the noisy score computed on the noisy counts $\tilde{v_i}$ as the weight for each edge. 

Obviously the DP option introduces considerable amount of noise in the graph construction process, and the original weight $w$ could be much different from the noisy counterpart $w'$, depending on the privacy parameter $\epsilon$ and the randomized generation of the Laplace noise. Also, the grid structure may impact the scale of errors, because the cosine similarity function takes the dot product of two noisy counts, and at least the error scales linearly with the number of grids $|g|$.  
From \lineref{line:hybrid_r_l} to \lineref{line:hybrid_r_u}, we compare the computed noisy weight $\tilde{w}$ to the lower bound threshold $r_l$ first. If it is smaller than $r_l$, we will directly prune away this edge. Then, we compare it to the upper bound threshold $r_u$. If $\tilde{w}$ is larger than $r_u$, we will directly insert the edge into the graph.

\lineref{line:hybrid_r_m} shows the borderline case, when we invoke the heavy secure version of the similarity function computation. The detailed steps of the secure computation are omitted, and the main idea is to feed the private inputs $D_i$ and $D_j$ and use secure primitives (including secure addition, multiplications, and if conditions) to guarantee that the private inputs are kept confidential during the entire process of computation.

\begin{proof}

  We reduce the Minimum Clique Partition (MCP) problem to the DOG problem. Since MCP problem is an NP-hard problem, DOG shares the same hardness. 

  We show that the DOG problem has a solution with variance $2\lambda/ \epsilon^2$ if and only if the MCP problem has a solution with $\lambda$ partitions. 

1. MCP $\Longrightarrow$ DOG:
If the MCP has a solution of $\lambda$ partitions, we simply return the optimal partition as the grouping for the DOG problem. Since the partitions in the MCP are disjoint, they form $\lambda$ valid groups $g_1, g_2, \ldots, g_\lambda$. Since the MCP solution returns a partition of cliques, each partition is also a clique, which satisfy the requirement that each group is a $r$-clique. Thus, the solution provides a valid grouping with $\lambda$ groups. Then, due to \lemref{lemma:minimize_lambda}, we know the solution provides a noisy answer with variance $2\lambda/ \epsilon^2$ to the FPRC problem. 

2. DOG $\Longrightarrow$ MCP:
If the DOG problem has a solution providing noisy answers with variance $2\lambda/ \epsilon^2$, by \lemref{lemma:minimize_lambda}, we know that the number of groups is $\lambda$. Since each group in $g_1, \ldots, g_\lambda$ is disjoint, and each group is a $r$-clique, as required \defref{def:dog_grouping_problem}, the solution of the grouping is a solution to the MCP problem, which provides a partition of disjoint cliques, and the number of cliques is $\lambda$. 

\end{proof}

\begin{example}
  We demonstrate the greedy solution using the example in \figref{fig:example_spatial_sim_graph} as the input. The original graph $G_s$ contains 4 nodes: $u_1$-$u_4$. It only contains 3 edges: $(u_1, u_2), (u_1, u_3)$ and $(u_2, u_3)$. The complement graph, $G_c=(V, E_c)$, has the same vertex set $V$ and 3 edges: $(u_1, u_4), (u_2, u_4), (u_3, u_4)$, which are all edges that do not exist in the original edge set $E$.

  First, we find vertex $u_1$ and assign it to the 1st group. \textit{assigned}$[1]=1$, and $\lambda=1$. 
  
  Then, we look at vertex $u_2$. $u_2$ has only 1 neighbor in $G_c$, which is $u_4$ because of edge $(u_2, u_4)$, but $u_4$ has not been assigned to any group yet. So, the smallest possible group id is 1. So $u_2$ is assigned to the 1st group as well. \textit{assigned}$[2]=1$, and $\lambda=1$. 

  Next, we continue to $u_3$. Similar to $u_2$, the only neighbor $u_3$ has is $u_4$ because of edge $(u_3, u_4)$. The smallest available group id is also 1. So, $u_3$ is assigned to the 1st group. \textit{assigned}$[3]=1$, and $\lambda=1$. 

  Last, we look at vertex $u_4$. All of the other vertexes are $u_4$'s neighbors, and all of them are using the same group id $gid=1$. Thus, it uses the smallest available group id $gid=2$, which is a new group id. Now, \textit{assigned}$[1]=2$, and $\lambda=2$. The returned results are $g_1, g_2$, where $g_1=\{u_1, u_2, u_3\}$ and $g_2=\{u_4\}$. The number of groups $\lambda=2$, indicating that we found 2 groups. 

\end{example}

\subsubsection{Constrained version}
\label{subsec:constrained_version}
As we have briefly discussed in key idea of FedGroup, each data owner may still desire privacy protection even within families (groups). Thus, we propose two variants of the Data Owner Grouping (DOG) problem, taking personal privacy requirements into consideration. 

\begin{figure}[t]
	\centering \vspace{0ex}
	   \begin{subfigure}[b]{0.4\textwidth}
		   \centering
		   \includegraphics[width=\textwidth]{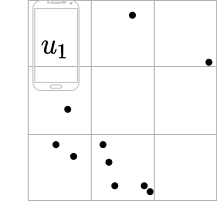}
		   \caption{Data silo $D_1$.}
		   \label{subfig:example_data_silo_d1}
	   \end{subfigure}
	   \hfill
		\begin{subfigure}[b]{0.4\textwidth}
		   \centering
		   \includegraphics[width=\textwidth]{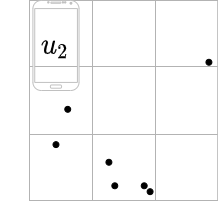}
		   \caption{Data silo $D_2$.}
		   \label{subfig:example_data_silo_d2}
	   \end{subfigure}
	
	\caption{\small Examples of two data silos $D_1$ and $D_2$.}	
	
	\label{fig:example_data_silo}
   \end{figure}

Note that we tradeoff the accuracy of the graph (as compared to the non-private version of constructing $G_s$, \aka the ground-truth graph) with the need to provide proven privacy guarantee. However our proposed hybrid solution strives an excellent balance between the utility/privacy tradeoff, as our experimental results show in \secref{sec:experiment}.

The lower bounds and the upper bound thresholds $r_l$ and $r_u$ could be derived analytically from the given privacy budget $\epsilon$ and the initial threshold $r$, given a confidence interval. In our experiments, we consider these two parameters as hyper-parameters, and our initial attempts have already provided satisfactory results.

Next we show the correctness of the greedy solution. 

\begin{theorem}
  The groups returned by \algoref{algo:greedy_grouping} are $r$-cliques and disjoint, which satisfy the requirements in the DOG problem (\defref{def:dog_grouping_problem}). 
\end{theorem}
\begin{proof}
  We prove that each group is a $r$-cliques by contradiction. Suppose that a group $g$ is not clique, \ie there exists $u_i, u_j\in g$ and there is no edge between them in the original graph $\tilde{G_s}$ (\ie $(u_i, u_j)\notin \tilde{E}$). So, in the complement graph of $\tilde{G_s}$, $G_c$, such an edge $(u_i, u_j)$ exits. However, according to the algorithm, at \lineref{line:greedy_nextgroupid}, we look for the smallest-possible group id that has not been used by any of the neighbors of a given node, which means it is impossible that $u_i$ and $u_j$ are assigned to the same group $g$. Thus, this leads to a contradiction. Each group in $g_1, g_2, \ldots, g_\lambda$ are $r$-cliques. 

  The groups are disjoint as we only assign a data owner $u$ to a group once, and after it is assigned, the assignment is not changed anymore. 

  Thus, the groups returned by \algoref{algo:greedy_grouping} are $r$-cliques and disjoint, which satisfy the requirements in the DOG problem (\defref{def:dog_grouping_problem}). 
\end{proof}

 The reason we call the solution \textit{greedy} is that, when we try to find which group that a data owner $u$ should be assigned to, it always finds the smallest available group id among the neighbors of $u$. If all group ids previously used are occupied by some neighbor of $u$, then we increment the total number of groups and use a new group id. The detailed steps of the solution are shown in \algoref{algo:greedy_grouping}.
 
 The algorithm uses an array called \texttt{assigned} to keep track of a mapping storing whether a vertex (a data owner) has been assigned or not. If the $i$-th data owner has been assigned to a group with id $k$, then \texttt{assigned}$[i]=k$. If the data owner has not been assigned to any group yet, \texttt{assigned}$[i]=0$. 

Each group $g_{gid}$ is a group (implemented as a set), where $gid$ is the id of the group. For example, $g_1$ contains all the data owners for the first group. $g_{gid}$ contains all the data owners for the group with an id of value $gid$. For an unassigned vertex $u$, at \lineref{line:greedy_nextgroupid} (the details are omitted in this pseudo-code), we iterate over the neighbors of $u$ and looks for the smallest possible group id not been used by any of the neighbors. Then, the smallest possible group id is used as the group assignment for vertex $u$. 

In the end, the assigned groups $g_1, g_2, \ldots, g_\lambda$ are returned. Each group contains the corresponding data owners in the group, and $\lambda$ denotes the number of groups found by the greedy solution. 

\fakeparagraph{Utility analysis}
Injecting Laplace in the similarity computation would inevitably result in loss of accuracy in terms of whether an edge exits or not in $G_s$. Because of the use of MPC technique as a supplementary method, we ensure that edge is correctly identified if the spatial similarity between two data silos are high. The results are shown in the experimental results in \secref{sec:experiment}.

\begin{theorem}
  The hybrid solution in \algoref{algo:construct_graph_hybrid} satisfies all privacy requirements (R1-R3) in \defref{def:problem}. 
\end{theorem}
\begin{proof}

We start with the analysis on the intermediate results released by each data silo $D_i$, and show that R2 is satisfied. For a data owner $u_i$, the only intermediate results released the is noisy count vector $\tilde{v_i}$, constructed from the private data silo $D_i$. The true counts vector $v_i$ is constructed in a way illustrated in \secref{subsubsec:spatial_sim_graph}, by simply taking the count of data records within the range of each grid in the grid structure $g$. $v_i$ is simply the vector of counts for each grid, with length $|g|$. The noisy count vector $\tilde{v_i}$ is obtained by injecting Laplace noise with scale $1/\epsilon$ to each grid. The change of one data record in data silo $D_i$ may change the count in $v_i$ by at most one entry, and it is at most changed by +/-1 (\aka the sensitivity is 1). Thus, regarding to the released intermediate results, the noisy counts vector $\tilde{v_i}$, injecting $\text{Lap}(1/\epsilon)$ suffices to provide $\epsilon$-DP. In conclusion, R2 is satisfied. 

Then, we show that R1 is satisfied by utilizing the post-processing and parallel composition property of DP. For \algoref{algo:construct_graph_hybrid}, the output is a constructed spatial similarity graph $G_s$, which is stored on the service provider $S$. There is no released final result $\tilde{Q}$ to the end user yet. Thus, to show R1 is satisfied, we show that the constructed $G_s$ satisfies $\epsilon$-DP, \wrt changing one record in any data silo $D_i$. Let us focus on a particular data silo $D_i$ for data owner $u_i$, and consider a neighboring database $D'_i$ with only one data record different from $D_i$. Then, $u_i$ releases the noisy count vectors $\tilde{v_i}$ to the service provider $S$. As we have shown above (R2 is satisfied), $\tilde{v_i}$ satisfies $\epsilon$-DP. Then, we focus on \algoref{algo:construct_graph_hybrid}. For \lineref{line:hybrid_r_l} to \lineref{line:hybrid_r_u}, only $\tilde{v_i}$ is used. Because of the post-processing properties of DP, we know any edge inserted to $E$ (or the constructed $G_s$) does not violate $\epsilon$-DP. For the borderline case at \lineref{line:hybrid_r_m}, the MPC protocol is used, thus the private database is strictly confidential. Thus, we conclude that the constructed $G_s$ satisfies $\epsilon$-DP for data silo $D_i$. Lastly, because data silos $D_1, \ldots, D_m$ are disjoint with each other, changing one record in any one of them does not affect the others's noisy counts, we know the constructed $G_s$ satisfies $\max_{i=1}^m \epsilon = \epsilon$-DP, by the parallel composition theorem of DP. Here, each data silo is using the same privacy requirement $\epsilon$. 

It is straightforward to see that R3 satisfies because of the use of MPC protocol. The only place that the private inputs are used is at \lineref{line:hybrid_r_m}, and the private inputs are fed to the MPC protocol, and thus the confidentiality of the data are provided.   

\end{proof}

\subsubsection{Constrained version of the DOG problem}
\label{subsec:constrained_version_appendix}
As we have briefly discussed in key idea of FedGroup, each data owner may still desire privacy protection even within families (groups). Thus, we propose two variants of the Data Owner Grouping (DOG) problem, taking personal privacy requirements into consideration.

in this section, we propose two variants of the Data Owner Grouping (DOG) problem, taking personal privacy requirements into consideration. 

\begin{definition}[DOG-Global]
  \label{def:dog_global_constraint}
  The DOG-Global problem has the same inputs and outputs as the DOG problem in \defref{def:dog_grouping_problem}, with an additional requirement:
  each data owner's assigned group has at most $t$ data owners, \ie the group size should not exceed $t$.  
\end{definition}

This variant of the DOG problem is called the DOG-Global problem, as it has a global constraint for all data owners. Each data owner is only willing to share about their private data with at most $t$ other data owners. 

\begin{definition}[DOG-Personal]
  \label{def:dog_personal_constraint}
  The DOG-Personal problem has the same inputs and outputs as the DOG problem in \defref{def:dog_grouping_problem}, with the following additional requirements for each data owner:
  each data owner $u$ has a personal privacy requirement $t_u$, and $u$ could be sharing their private data with at most $t_u$ data owners, \ie the size of the group that $u$ is assigned to should not exceed $t_u$. 
\end{definition}

Different from the global constraint, which asks that every data owner shares the same privacy constraint $t$, the DOG-Personal variant allows each data owner to have a personal constraint $t_u$. 

We modify the greedy solution for the DOG problem to counter the privacy constraints (both the DOG-Global and the Personal variants). 

For the DOG-Global variant, the modification of the greedy solution is that whenever we assign a data owner $u$ to a potential group, we check the size of the group and avoid assigning the data owner to the group if the size reaches $t$. 

For the personal variant, the modification is that when we attempt to assign data owner $u$ to a potential group with id $i$, we check whether $|g_i| \geq t_u$, \ie the size of group $i$ is greater or equal to $t_u$. If that is the case, then the assignment would make $g_i$ become $t_u+1$, which violates the personal constraint of $u$. So, we pick the next smallest-numbered available group id as a candidate group for data owner $u$. 

Due to the space limit, we omit the detailed pseudo-code for the modified greedy solution. The correctness of the solutions are straightforward as the constrained versions only differ with the original DOG problem in the constraints part. 

\section{Additional experimental results}

\subsection{System configuration}
 The experiment is conducted on a CentOS Linux system with Intel(R) Xeon(R) Gold 6240R CPU @2.40GHz and 1007G memory. We implement the methods and conduct the experiments using Python and C. The MPC extension is provided by Obliv-C \cite{zahur2015obliv}. We use two processes on the same machine to simulate the client-server MPC protocol.
 
 To get a better understanding of the behavior of our proposed solution FG, we test the query accuracy when using the ground-truth spatial similarity graphs constructed from data owners' inputs and directly perform the grouping and query answer aggregation in FG. The tested method is called FG-C (where C stands for \underline{C}lear/ground-truth graph). For the clique listing variant FG-Ex, we also test it on the clear graph and use the name FG-ExC. The results are shown in \figref{subfig:exp_MRE_eps_C_vs_N} and \figref{subfig:exp_MAE_eps_C_vs_N}. Obviously, due to the privacy/utility tradeoff, FG-C is better than FG, and FG-ExC is slightly better than FG-Ex, because using the noisy graph (as opposed to the original graph) leads to a slightly worse MRE and MAE.  However, FG offers a similar level of MRE and MAE as compared to FG-ExC and FG-Ex, which shows that our greedy heuristics are effective as compared to other alternatives. 
 
 \fakeparagraph{Noisy vs. Original graph for different $m$}
 Similar to the observations we made for different $\epsilon$, FG provides a slightly worse MRE and MAE than the case when we directly use the ground-truth spatial similarity graph (\figref{subfig:exp_MRE_m_C_vs_N}), due to an inevitable utility/privacy tradeoff. However, the margin is small, and FG constantly beats the other heuristic approach FG-Ex.

\end{document}